\documentclass[a4paper,11pt]{article}

\usepackage[T1]{fontenc}
\usepackage[utf8]{inputenc}
\usepackage[english]{babel}
\usepackage{url}
\usepackage{fullpage}
\usepackage{setspace}
\usepackage{times}
\usepackage{xspace}

\usepackage{listings}
\usepackage{amsmath}
\usepackage{amsthm}
\usepackage{amssymb}
\usepackage{amsfonts}
\usepackage{graphicx}
\usepackage{verbatim}
\usepackage{stmaryrd}
\usepackage{enumitem}
\usepackage{mathtools}
\usepackage{hyperref}
\usepackage[all,cmtip]{xy}
\usepackage{relsize}
\usepackage{appendix}
\usepackage{tabularx}

\newdimen\proofrulebreadth \proofrulebreadth=.05em
\newdimen\proofdotseparation \proofdotseparation=1.25ex
\newdimen\proofrulebaseline \proofrulebaseline=2ex
\newcount\proofdotnumber \proofdotnumber=3
\let\then\relax
\def\hfi{\hskip0pt plus.0001fil}
\mathchardef\squigto="3A3B
\newif\ifinsideprooftree\insideprooftreefalse
\newif\ifonleftofproofrule\onleftofproofrulefalse
\newif\ifproofdots\proofdotsfalse
\newif\ifdoubleproof\doubleprooffalse
\let\wereinproofbit\relax
\newdimen\shortenproofleft
\newdimen\shortenproofright
\newdimen\proofbelowshift
\newbox\proofabove
\newbox\proofbelow
\newbox\proofrulename
\def\shiftproofbelow{\let\next\relax\afterassignment\setshiftproofbelow\dimen0 }
\def\shiftproofbelowneg{\def\next{\multiply\dimen0 by-1 }%
\afterassignment\setshiftproofbelow\dimen0 }
\def\setshiftproofbelow{\next\proofbelowshift=\dimen0 }
\def\setproofrulebreadth{\proofrulebreadth}

\def\prooftree{
\ifnum  \lastpenalty=1
\then   \unpenalty
\else   \onleftofproofrulefalse
\fi
\ifonleftofproofrule
\else   \ifinsideprooftree
        \then   \hskip.5em plus1fil
        \fi
\fi
\bgroup
\setbox\proofbelow=\hbox{}\setbox\proofrulename=\hbox{}%
\let\justifies\proofover\let\leadsto\proofoverdots\let\Justifies\proofoverdbl
\let\using\proofusing\let\[\prooftree
\ifinsideprooftree\let\]\endprooftree\fi
\proofdotsfalse\doubleprooffalse
\let\thickness\setproofrulebreadth
\let\shiftright\shiftproofbelow \let\shift\shiftproofbelow
\let\shiftleft\shiftproofbelowneg
\let\ifwasinsideprooftree\ifinsideprooftree
\insideprooftreetrue
\setbox\proofabove=\hbox\bgroup$\displaystyle 
\let\wereinproofbit\prooftree
\shortenproofleft=0pt \shortenproofright=0pt \proofbelowshift=0pt
\onleftofproofruletrue\penalty1
}

\def\eproofbit{
\ifx    \wereinproofbit\prooftree
\then   \ifcase \lastpenalty
        \then   \shortenproofright=0pt  
        \or     \unpenalty\hfil         
        \or     \unpenalty\unskip       
        \else   \shortenproofright=0pt  
        \fi
\fi
\global\dimen0=\shortenproofleft
\global\dimen1=\shortenproofright
\global\dimen2=\proofrulebreadth
\global\dimen3=\proofbelowshift
\global\dimen4=\proofdotseparation
\global\count255=\proofdotnumber
$\egroup  
\shortenproofleft=\dimen0
\shortenproofright=\dimen1
\proofrulebreadth=\dimen2
\proofbelowshift=\dimen3
\proofdotseparation=\dimen4
\proofdotnumber=\count255
}

\def\proofover{
\eproofbit 
\setbox\proofbelow=\hbox\bgroup 
\let\wereinproofbit\proofover
$\displaystyle
}%
\def\proofoverdbl{
\eproofbit 
\doubleprooftrue
\setbox\proofbelow=\hbox\bgroup 
\let\wereinproofbit\proofoverdbl
$\displaystyle
}%
\def\proofoverdots{
\eproofbit 
\proofdotstrue
\setbox\proofbelow=\hbox\bgroup 
\let\wereinproofbit\proofoverdots
$\displaystyle
}%
\def\proofusing{
\eproofbit 
\setbox\proofrulename=\hbox\bgroup 
\let\wereinproofbit\proofusing
\kern0.3em$
}

\def\endprooftree{
\eproofbit 
  \dimen5 =0pt
\dimen0=\wd\proofabove \advance\dimen0-\shortenproofleft
\advance\dimen0-\shortenproofright
\dimen1=.5\dimen0 \advance\dimen1-.5\wd\proofbelow
\dimen4=\dimen1
\advance\dimen1\proofbelowshift \advance\dimen4-\proofbelowshift
\ifdim  \dimen1<0pt
\then   \advance\shortenproofleft\dimen1
        \advance\dimen0-\dimen1
        \dimen1=0pt
        \ifdim  \shortenproofleft<0pt
        \then   \setbox\proofabove=\hbox{%
                        \kern-\shortenproofleft\unhbox\proofabove}%
                \shortenproofleft=0pt
        \fi
\fi
\ifdim  \dimen4<0pt
\then   \advance\shortenproofright\dimen4
        \advance\dimen0-\dimen4
        \dimen4=0pt
\fi
\ifdim  \shortenproofright<\wd\proofrulename
\then   \shortenproofright=\wd\proofrulename
\fi
\dimen2=\shortenproofleft \advance\dimen2 by\dimen1
\dimen3=\shortenproofright\advance\dimen3 by\dimen4
\ifproofdots
\then
        \dimen6=\shortenproofleft \advance\dimen6 .5\dimen0
        \setbox1=\vbox to\proofdotseparation{\vss\hbox{$\cdot$}\vss}%
        \setbox0=\hbox{%
                \advance\dimen6-.5\wd1
                \kern\dimen6
                $\vcenter to\proofdotnumber\proofdotseparation
                        {\leaders\box1\vfill}$%
                \unhbox\proofrulename}%
\else   \dimen6=\fontdimen22\the\textfont2 
        \dimen7=\dimen6
        \advance\dimen6by.5\proofrulebreadth
        \advance\dimen7by-.5\proofrulebreadth
        \setbox0=\hbox{%
                \kern\shortenproofleft
                \ifdoubleproof
                \then   \hbox to\dimen0{%
                        $\mathsurround0pt\mathord=\mkern-6mu%
                        \cleaders\hbox{$\mkern-2mu=\mkern-2mu$}\hfill
                        \mkern-6mu\mathord=$}%
                \else   \vrule height\dimen6 depth-\dimen7 width\dimen0
                \fi
                \unhbox\proofrulename}%
        \ht0=\dimen6 \dp0=-\dimen7
\fi
\let\doll\relax
\ifwasinsideprooftree
\then   \let\VBOX\vbox
\else   \ifmmode\else$\let\doll=$\fi
        \let\VBOX\vcenter
\fi
\VBOX   {\baselineskip\proofrulebaseline \lineskip.2ex
        \expandafter\lineskiplimit\ifproofdots0ex\else-0.6ex\fi
        \hbox   spread\dimen5   {\hfi\unhbox\proofabove\hfi}%
        \hbox{\box0}%
        \hbox   {\kern\dimen2 \box\proofbelow}}\doll%
\global\dimen2=\dimen2
\global\dimen3=\dimen3
\egroup 
\ifonleftofproofrule
\then   \shortenproofleft=\dimen2
\fi
\shortenproofright=\dimen3
\onleftofproofrulefalse
\ifinsideprooftree
\then   \hskip.5em plus 1fil \penalty2
\fi
}

Date: Tue, 19 May 1998 16:45:32 +0100
From: Simon Gay <simon@dcs.rhbnc.ac.uk>

I've got another problem when combining
your packages with elsart.cls. The code

\documentclass{elsart}

\input{diagrams}
\begin{document}

\[
\begin{prooftree}
A \rTo^{f} B
\justifies
p \rTo^{(a,c)} q
\end{prooftree}
\]

\end{document}

doesn't leave enough space below the line in the proof tree, so that
the (a,c) label on the lower arrow runs into the line. It's fine with
article.cls.

\urlstyle{sf}
\onehalfspacing
\setdescription{leftmargin=\parindent,labelindent=\parindent}
\lstset{frame=shadowbox}


\newtheorem{mytheo}{Theorem}[section]
\newtheorem{myprop}[mytheo]{Proposition}
\newtheorem{mycoro}[mytheo]{Corollary}
\newtheorem{mylem}[mytheo]{Lemma}

\theoremstyle{definition}
\newtheorem{mydef}[mytheo]{Definition}

\theoremstyle{remark}
\newtheorem{mynot}[mytheo]{Notation}
\newtheorem{mysam}[mytheo]{Example}
\newtheorem{myrem}[mytheo]{Remark}
\newtheorem{myprob}[mytheo]{Open problem}

\everymath=\expandafter{\the\everymath\displaystyle}

\newcommand{\bra}[1]{\left\langle #1 \right|}
\newcommand{\ket}[1]{\left| #1 \right\rangle}
\newcommand{\abs}[1]{\left| #1 \right|}
\newcommand{\demi}{\frac{1}{2}}
\newcommand{\IC}{\textsc{I}}
\newcommand{\vect}[1]{\overrightarrow{#1}}
\newcommand{\C}{\mathbb{C}}
\newcommand{\R}{\mathbb{R}}
\newcommand{\N}{\mathbb{N}}
\newcommand{\Z}{\mathbb{Z}}
\newcommand{\U}{\mathbb{U}}
\newcommand{\PX}{\mathcal{P}(X)}
\newcommand{\B}{\mathcal{B}}
\newcommand{\Proj}{\operatorname{Proj}}
\newcommand{\deriv}{\mathrm{d}}
\newcommand{\prodscal}[2]{\left\langle #1 \vert #2 \right\rangle}
\newcommand{\scprod}{\mathrel{\bullet}}
\newcommand{\norm}[1]{\left\| #1 \right\|}
\newcommand{\set}[1]{\left\lbrace #1 \right\rbrace}
\newcommand{\pair}[2]{\left\langle #1 , #2 \right\rangle}
\newcommand{\ltwo}{\ell^2}
\newcommand{\unit}{[0,1]}
\newcommand{\linspan}{\operatorname{span}}
\renewcommand{\Im}{\operatorname{Im}}
\renewcommand{\Re}{\operatorname{Re}}
\newcommand{\kl}[1]{{\mathcal Kl}(#1)}
\newcommand{\kln}[1]{{\mathcal Kl_{\N}}(#1)}
\newcommand{\EiMo}{\mathcal{EM}}
\newcommand{\Ef}{{\mathcal Ef }}
\newcommand{\lwnr}{\sqsubseteq}
\newcommand{\fix}{\operatorname{fix}}
\newcommand{\dirsubset}{\subseteq_{dir}}
\newcommand{\lub}{\bigvee}
\newcommand{\glb}{\bigwedge}
\newcommand{\ran}{\operatorname{ran}}
\newcommand{\wkp}{\operatorname{wp}}

\newcommand{\Pred}{\operatorname{Pred}}
\newcommand{\PredS}{\operatorname{Pred_{\leq 1}}}

\newcommand{\Komp}{\operatorname{\mathcal{K}}}
\newcommand{\Atom}{\operatorname{\mathcal{A}}}

\newcommand{\lublambda}{\lub_{\lambda \in \Lambda}}
\newcommand{\lubgamma}{\lub_{\gamma \in \Gamma}}

\newcommand{\Mb}{\mathcal M}
\newcommand{\DM}{\mathcal D_{= 1}}
\newcommand{\SDM}{\mathcal D_{\leq 1}}
\newcommand{\RM}{\mathcal R_{= 1}}
\newcommand{\SRM}{\mathcal R_{\leq 1}}
\newcommand{\EM}{\mathcal E_{= 1}}
\newcommand{\SEM}{\mathcal E_{\leq 1}}
\newcommand{\EA}{\mathbf{EA}}
\newcommand{\EAS}{\EA_\mathbf{s}}
\newcommand{\EMod}{\mathbf{EMod}}
\newcommand{\EModS}{\EMod_\mathbf{s}}
\newcommand{\GEA}{\mathbf{GEA}}
\newcommand{\GEMod}{\mathbf{GEMod}}
\newcommand{\sdcGEMod}{\mathbf{sdcGEMod}}
\newcommand{\dcEModS}{\mathbf{dc}\EModS}
\newcommand{\sdcEModS}{\mathbf{s}\dcEModS}
\newcommand{\Conv}{\mathbf{Conv}}
\newcommand{\SConv}{\mathbf{SubConv}}
\newcommand{\Dcpo}{\mathbf{Dcpo}}

\newcommand{\FdHilb}{\mathbf{FdHilb}}
\newcommand{\CStar}{\mathbf{CStar}}
\newcommand{\WStar}{\mathbf{WStar}}
\newcommand{\MIU}{\CStar_\mathrm{MIU}}
\newcommand{\PU}{\CStar_\mathrm{PU}}
\newcommand{\PSU}{\CStar_\mathrm{PsU}}
\newcommand{\wMIU}{\WStar_\mathrm{MIU}}
\newcommand{\wPU}{\WStar_\mathrm{NU}}
\newcommand{\wPSU}{\WStar_\mathrm{NsU}}
\newcommand{\CPSU}{\mathbf{C}\PSU}
\newcommand{\FdCPSU}{\mathbf{FdC}\PSU}
\newcommand{\NS}{\mathcal{NS}}

\newcommand{\opp}[1]{#1^\mathbf{op}}
\newcommand{\sa}[1]{#1_{\mathrm{sa}}}

\newcommand{\guillemets}[1]{\og #1 \fg{}}
\newcommand{\supernote}[1]{$^\text{#1}$}
\newcommand{\supercite}[1]{\supernote{\cite{#1}}}

\title{On operator algebras in quantum computation}
\author{Mathys Rennela, under the supervision of Bart Jacobs \\ Institute for Computing and Information Sciences, Radboud Universiteit Nijmegen)}
\date{}

\begin{document}

\maketitle

\subsubsection*{The general context}
In the following we discuss how the theory of operator algebras, also called operator theory, can be applied in quantum computer science. From a computer scientist point of view, we will discuss some fundamental results of operator theory and their relevance in the context of domain theory.
The theory of operator algebras originated in functional analysis in the 1930s and was extensively applied in mathematical physics, in order to understand the mathematical foundations of quantum mechanics. In the past 15 years, domain theory was successfully applied to quantum computation, for semantics and verification (see \cite{gay} for an overview of the literature). 

\subsubsection*{The research problem}
Our aim here is to use the theory of operator algebras to study the differences and similarities between probabilistic and quantum computations, by unveiling their domain-theoretic and topological structure. To our knowledge, the deep connection between the theory of operator algebras and domain theory was not fully exploited before. This might be due to the fact that the theory of operator algebras, mostly unknown to computer scientists, was developed way before the theory of domains.

Although there is now a real commercialization of quantum cryptographic systems, it is still unknown if quantum computers exist and moreover, it is yet unknown what such a computer (if any) would look like. However, in our opinion, it is important to provide some formal tools for the design and verification at an early level, to prevent system failure.

\subsubsection*{Our contribution}
Our main contribution is a connection between two different communities: the community of theoretical computer scientists, which use domain theory to study program language semantics (and logic), and the community of mathematicians and theoretical physicists, which use a special class of algebras called W*-algebras to study quantum mechanics. This connection involves a quantum domain theory and a quantum setting for a weakest precondition calculus, described categorically.  We will also introduce the notion of effect algebras, in order to associate a predicate logic to computations.

\subsubsection*{Arguments supporting its validity}
We only assume that W*-algebras are suitable for representing quantum computations, which is a common assumption in mathematical physics. During our research, we have found in the standard literature of W*-algebras some theorems suggesting that W*-algebras can be successfully used to provide a semantics for quantum computations, although it was not explicitly expressed in domain-theoretic terms. 


\subsubsection*{Future work}
Further research will concentrate on applying these brand new settings to the formal representation of quantum cryptographic protocols in a computer algebra tool, for verification and simulation.


Peter Selinger gave a denotational semantics of a quantum programming language QPL, which features loops, recursive procedures and structured data types $\cite{selinger}$. In this semantics, the type of bits and qubits is defined by bounded operators on finite-dimensional Hilbert spaces, which are W*-algebras. Thus, the denotational semantics of QPL do not consider at all infinite-dimensional W*-algebras. It turns out that quantum streams (i.e. infinite sequences of qubits) can be semantically denoted as infinite tensor products of W*-algebras, which are necessarily infinite-dimensional W*-algebras, see \cite[III.3.1.4]{blackadar}. This is an interesting point of start for future research.

\tableofcontents

\newpage

\section{Preliminaries}
From now on, we will assume that the reader is familiar with category theory. Otherwise, an introduction to category theory can be found in \cite{awodey,maclane}.

In this section, assuming basic knowledge of linear algebra, we will briefly recall standard notions of topology and order theory, and then lay the foundation for further discussion on quantum computation, providing some standards definitions in the theory of operator algebras. The interested reader will find in Appendix \ref{appendix-correspondence} a detailed correspondence between operator theory and order theory.
 
\subsection{Order theory}
\begin{mydef}
A set $P$ together with a partial order $\leq$ is called a partially ordered set (or poset).

A bottom of $P$ is an element $\perp \in P$ such that $\perp \leq x$ for every $x \in P$.

A top of $P$ is an element $\top \in P$ such that $x \leq \top$ for every $x \in P$.

A bounded poset is a poset that has both a top and a bottom.
\end{mydef}

For every poset, it is clear that if a top (or a bottom) exist, then it is unique.

\begin{mydef}
In a poset $P$, the down set of an element $x \in P$ is the set 
\[
  {\downarrow} x = \set{y \in P \mid y \leq x}.
\]
\end{mydef}

\begin{mydef}
A poset $(P,\leq)$ is a chain if every pair of elements of $P$ is comparable: 
\[
 \forall x, y \in P, x \leq y \text{ or } y \leq x.
\]
\end{mydef}

We denote respectively $a \vee b$ and $a \wedge b$ the least upper bound (or join) and the greatest lower bound (or meet) of two elements $a$ and $b$ of a poset, if they exist.
For any subset $X$, the join (resp. the meet) of $X$ is denoted by $\bigvee X$ (resp. denoted by $\bigwedge X$). 

\begin{mydef}
A non-empty subset $\Delta$ of a poset $P$ is called directed if every pair of elements of $\Delta$ has an upper bound. We denote it by $\Delta \dirsubset P$.
\end{mydef}

\begin{mydef}[Completeness]
Let $P$ be a poset.
\begin{itemize}
 \item $P$ is a directed-complete partial order (dcpo) if each directed subset has a least upper bound.
 \item $P$ is bounded-complete if for each subset $S \subseteq P$, $S$ has some upper bound implies that $S$ has a least upper bound.
 \item $P$ is chain-complete if all chains in $P$ have a least upper bound. 
\end{itemize}
\end{mydef}

It can be proved with Zorn's lemma that a poset is chain-complete if and only if it is directed-complete.

\begin{mydef}
Let $\phi : P \to Q$ be a function between two posets $P$ and $Q$.

$\phi$ is monotonic if $x \leq_P y \implies \phi(x) \leq_Q \phi(y)$ for all $x, y \in P$.

$\phi$ is Scott-continuous if for every directed subset $\Delta \dirsubset P$ with least upper bound $\lub \Delta \in P$, the subset $\phi(\Delta)$ of $Q$ is directed with least upper bound  $\bigvee \phi(\Delta) = \phi(\bigvee \Delta)$.

The set of all Scott-continuous maps from $P$ to $Q$ is denoted by $[P \rightarrow Q]$ and can be ordered pointwise: \[
 f \leq g \text{ if and only if } \forall x \in P, f(x) \leq_Q g(x) \quad (f, g \in [P \rightarrow Q])
\]

We denote by $\Dcpo$ the category with dcpos as objects and Scott-continuous maps as morphisms.
\end{mydef}

\begin{mytheo}[\cite{davey-priestley}, Theorem 8.9]\label{cont-dcpo}
Let $P$ and $Q$ be two posets.

The poset $[P \rightarrow Q]$ is a dcpo whenever $P$ and $Q$ are dcpos. 
\end{mytheo}

\begin{mydef}
Let $P$ and $Q$ be two posets with bottoms $\perp_P$ and $\perp_Q$ respectively.

A function $\phi : P \to Q$ is strict if $\phi(\perp_P)=\perp_Q$.

We denote by $\Dcpo_\perp$ (resp. $\Dcpo_{\perp !}$) the category of dcpos with bottoms and Scott-continuous maps (resp. strict Scott-continuous maps). 
\end{mydef}

The product $D_1 \times \cdots \times D_n$ of a family of dcpos $D_1, \cdots, D_n$ is defined by the $n$-tuples $(x_1,\cdots,x_n)$ where $x_i \in D_i$ for every $1 \leq i \leq n$. The partial order is defined coordinatewise by $(x_1,\cdots,x_n) \leq (y_1,\cdots,y_n)$ iff $x_i \leq y_i$ for every $1 \leq i \leq n$. It is known that the product of dcpos forms itself a dcpo where the least upper bounds are calculated coordinatewise. Moreover, the categories $\Dcpo$ and $\Dcpo_\perp$ are cartesian closed, whereas $\Dcpo_{\perp!}$ is only monoidal closed, see \cite{abramsky-jung}.

\subsection{Topology}
\begin{mydef}
Let $X$ be a nonempty set. A topology on $X$ is a subset $\tau$ of $\PX$ such that:
\begin{itemize}
 \item $X$ and $\emptyset$ are in $\tau$.
 \item If two sets $U$ and $V$ are in $\tau$, then $U \cap V$ is in $\tau$.
 \item If $I$ is the index set of a family $(U_i)_{i \in I}$ of elements of $\tau$, then $\bigcup_{i \in I} U_i \in \tau$.
\end{itemize}

A topological space $(X,\tau)$ is a set $X$ with a family $\tau$ that satisfies these properties. It is common to say that $X$ is a topological space when $\tau$ is understood from the context.
 
The elements of $X$ are called points and the elements of $\tau$ are called open sets.
\end{mydef}

\begin{mydef}
Let $(X,\tau)$ be a topological space. 

A subbase for $\tau$ is a subcollection $B$ of $\tau$ which generates $\tau$. That is to say, $\tau$ is the smallest topology which contains $B$: if a topology $\tau'$ on $X$ contains $B$, then it also contains $\tau$.

A subset $Y \subseteq X$ is a closed set if $X \setminus Y$ is an open set.

A subset $Y \subseteq X$ is a subspace of $X$ if $(Y,\tau')$ is a topological space, where $\tau'=\set{U \cap Y \mid U \in \tau}$.

A net is a function $(x_\lambda)_{\lambda \in \Lambda}$ from some directed set $\Lambda$ to $X$. 
\end{mydef}

A net is a generalization of the notion of sequence, which can be seen as a net with $A=\N$. Nets are used in topology to consider continuity for functions between topological spaces, since sequences do not fully encode all the information about such functions. In fact, the range of a function between topological spaces is not always the natural numbers but can be any topological space.

In a topological space $X$, a neighbourhood of a point $x \in X$ is a subset $V$ of $X$ such that $x \in U \subseteq V$ where $U$ is an open set of $X$. A (countable) basis $\mathcal{B}_x$ at a point $x \in X$ is a collection of (countable) neighbourhoods of $x$ such that, for every neighbourhood $V$ of $x$, there is a neighbourhood $V'$ in $\mathcal{B}_x$ such that $V' \subseteq V$. $X$ is said to be first-countable when every point has a countable basis. When $X$ is not first-countable, there might be some points $x \in X$ with an uncoutable basis. It follows that sequences, which are countable by definition, might not succeed to get close enough to $x$.

Indeed, a function $f : X \to Y$ between topological spaces is continuous at a point $x \in X$ if and only if, for every net $(x_\lambda)_{\lambda \in \Lambda}$ with $\lim x_\lambda = x$, we have $\lim f(x_\lambda) = f(x)$ \cite{willard}. This statement is generally not true if we replace "net" by "sequence", since we have to allow for more directed sets than just the natural numbers when $X$ is not first-countable.

\subsection{C*-algebras} 
\begin{mydef}

A Banach space is a normed vector space where every Cauchy sequence converges.\\
A Banach algebra is a linear associative algebra $A$ over $\C$ with a norm $\norm{\cdot}$ such that: \begin{itemize}
 \item The norm $\norm{\cdot}$ is submultiplicative: $\forall x, y \in A, \norm{xy} \leq \norm{x}\norm{y}$
 \item $A$ is a Banach space under the norm $\norm{\cdot}$.
\end{itemize}
\end{mydef}

\begin{mydef}
A unit is an element of a Banach algebra $A$ such that $a1=1a=a$ for every element $a \in A$.

A Banach algebra $A$ is unital if it has a unit $1$ and satifies $\norm{1}=1$.
\end{mydef}

\begin{mydef}
A *-algebra is a linear associative algebra $A$ over $\C$ with an operation $(-)^* : A \rightarrow A$ such that for all $x$, $y$ in $A$:
\begin{align*}
 (x^*)^*=x && (x+y)^*=(x^*+y^*) && (xy)^*=y^*x^* && (\lambda x)^* = \overline{\lambda}x^* \quad  (\lambda \in \C)
\end{align*}

A *-homomorphism of *-algebras is a linear map that preserves all this structure.
\end{mydef}

\begin{mydef}
A C*-algebra is a Banach *-algebra $A$ such that $\norm{x^*x} = \norm{x}^2$ for all $x \in A$.

This identity is sometimes called the C*-identity, and implies that every element $x$ of a C*-algebra is such that $\norm{x} = \norm{x^*}$.
\end{mydef}

We will now assume that C*-algebras are unital (i.e. have a unit element denoted $1$).

\begin{mydef}
Let $A$ be a C*-algebra.
\begin{itemize}
 \item An element $x \in A$ is self-adjoint (or hermetian) if $x=x^*$. \\
       We write $\sa{A} \hookrightarrow A$ for the subset of self-adjoint elements of $A$.
 \item An element $x \in A$ is positive if it can be written in the form $x=y^* y$, where $y \in A$. \\
       We write $A^+ \hookrightarrow A$ for the subset of positive elements of $A$.
\end{itemize}
\end{mydef}

Every self-adjoint element of a C*-algebra $A$ can be written as difference $x = x_+ - x_-$ where $x_+,x_- \in A^+$, with $\norm{x_+}, \norm{x_-} \leq \norm{x}$. Moreover, an arbitrary element $x$ of a C*-algebra $A$ can be written as linear combination of four positive elements $x_i \in A$ such that $x = x_1 - x_2 + i x_3 - i x_4$ with $\norm{x_i} \leq \norm{x}$, see \cite[II.3.1.2]{blackadar}.

\begin{mydef}
A linear map $f : A \rightarrow B$ of C*-algebras is a positive map if it preserves positive elements, i.e. $\forall x \in A^+, f(x) \in B^+$.

This means that $f$ restricts to a function $A^+ \to B^+$. Alternatively, $\forall x \in A, \exists y \in B,  f(x^*x) = y^*y$.
\end{mydef}

For every C*-algebra, the subset of positive elements is a convex cone and thus induces a partial order structure on self-adjoint elements, see \cite[Definition 6.12]{takesaki1}. That is to say, one can define a partial order on self-adjoint elements of a C*-algebra $A$ as the binary relation $\leq$ defined for $x, y \in \sa{A}$ by: $x \leq y$ if and only if $y-x \in A^+$. By convention, one writes $x \geq 0$ when $x \in A^+$.

\begin{mylem}\label{presrv-order}
A positive map of C*-algebras preserves the order $\leq$ on self-adjoint elements.
\end{mylem}

\begin{proof}
Let $f : A \to B$ be a positive map of C*-algebras and $x, y \in \sa{A}$.

If $x \leq y$, then $y-x \in A^+$. Thus $f(y)-f(x)=f(y-x) \in B^+$ and therefore, $f(x) \leq f(y)$.
\end{proof}


\begin{mydef}
Let $f : A \rightarrow B$ be a linear map between unital C*-algebras $A$ and $B$.

\begin{itemize}
 \item $f$ is a multiplicative map if $\forall x,y \in A, f(xy)=f(x)f(y)$.
 \item $f$ is an involutive map if $\forall x \in A, f(x^*)=f(x)^*$.
 \item $f$ is a unital map if it preserves the unit, i.e. $f(1) = 1$. 
 \item $f$ is a sub-unital map if $0 \leq f(1) \leq 1$.
\end{itemize}
\end{mydef}

\begin{mydef}
We shall define three categories $\MIU$, $\PU$ and $\PSU$ with C*-algebras as objects but different morphisms: 

\begin{itemize}
 \item A morphism $f : A \to B$ in $\MIU$ is a \textbf{M}ultiplicative \textbf{I}nvolutive \textbf{U}nital map (or MIU-map).
 \item A morphism $f : A \to B$ in $\PU$ is a \textbf{P}ositive \textbf{U}nital map (or PU-map).
 \item A morphism $f : A \to B$ in $\PSU$ is a \textbf{P}ositive \textbf{s}ub-\textbf{U}nital map (or PsU-map). 
\end{itemize}
\end{mydef}

\begin{mylem}
There are inclusions $\MIU \hookrightarrow \PU \hookrightarrow \PSU$.
\end{mylem}

\begin{proof}
Let $f : A \to B$ be a linear map between two C*-algebras $A$ and $B$.

If $f : A \to B$ is a MIU-map, then for every $x \in A$, $f(x^*x)=f(x^*)f(x)=(f(x))^*f(x)$. It follows that $f$ is positive ($\forall x \in A$, $f(x^*x)=y^*y$ where $y=f(x)\in B$).

If $f : A \to B$ is a PU-map, $f(1)=1$ and therefore $0 \leq f(1) \leq 1$. Hence, $f$ is a PsU-map.
\end{proof}

\begin{mydef}
A state on a C*-algebra $A$ is a PsU-map $\phi : A \to \C$. 
The state space of a C*-algebra $A$ is the hom-set $\PSU(A,\C)$.
\end{mydef}

\subsection{Hilbert spaces}
\begin{mydef}
A Hilbert space is a Banach space $H$ together with an inner product $\prodscal{\cdot}{\cdot}$ and a norm defined by $\norm{x}^2=\prodscal{x}{x}$ ($x \in H$).
\end{mydef}

\begin{myprop}[Cauchy-Schwarz inequality]
Let $H$ be a Hilbert space. 

For every $x, y \in H$, $\abs{\prodscal{x}{y}} \leq \norm{x}\norm{y}$.
\end{myprop}

We now consider the situation of operators (i.e. linear maps) $H \rightarrow H$ on a Hilbert space $H$.

\begin{mydef}
A linear map $f : A \rightarrow B$ between Banach spaces is a bounded operator if there exists a $k > 0$ such that $\norm{f(a)}_B \leq k \cdot \norm{a}_A$ for every $a$ of $A$. 

The collection of all bounded operators between two Hilbert spaces $H_1$ and $H_2$ is denoted $\B(H_1,H_2)$. For every Hilbert space $H$, we denote by $\B(H)$ the collection $\B(H,H)$.

The set of effects $\Ef(H)$ on a Hilbert space $H$ is the set of positive bounded operators below the unit, i.e. $\Ef(H)=\set{T \in \B(H) \mid 0 \leq T \leq 1}$.
\end{mydef}

\begin{mydef}
Let $H$ be a Hilbert space.
For every bounded operator $T \in B(H)$, we define the following sets:
\begin{align*}
 \text{\textbf{Kernel:}} \quad \ker T = \set{x \in H \mid Tx = 0} &&
 \text{\textbf{Range:}} \quad \ran T = \set{y \in H \mid \exists x \in H, y = T x}
\end{align*}
\end{mydef}

For every Hilbert space $H$, it is known that $\B(H)$ is a Banach space and therefore a C*-algebra. Self-adjoint and positive elements of $\B(H)$ can be defined alternatively through the inner product of $H$, as shown by the two following theorems\footnote{We will deliberately admit these standard theorems, as their proofs involve arguments coming from spectral theory, which is totally out of our scope. For more details, we refer the reader to \cite[II.2.12,VIII.3.8]{conway}.}, taken from \cite{conway}:

\begin{mytheo}\label{self-adjoint-BH}
Let $H$ be a Hilbert space and $T \in \B(H)$. 

Then $T$ is self-adjoint if and only if $\forall x \in H, \prodscal{Tx}{x} \in \R$.
\end{mytheo}

\begin{mytheo}\label{positive-BH}
Let $H$ be a Hilbert space and $T \in \B(H)$. 

Then $T$ is positive if and only if $T$ is self-adjoint and $\forall x \in H, \prodscal{Tx}{x} \geq 0$.
\end{mytheo}

\subsection{W*-algebras}
In this section, we investigate some topological structures of bounded operators on Hilbert spaces, in order to define a special class of C*-algebras, known as W*-algebras (or von Neumann algebras), that were introduced in the 1930s and 1940s in a series of papers by Murray and von Neumann \cite{murray-von-neumann}, and latter used by Girard for his Geometry of Interaction \cite{girard-GoI5}.

There are several standard topologies that one can define on $\B(H)$ (see \cite{takesaki1,blackadar} for an overview).
 
\begin{mydef}
The operator norm $\norm{T}$ is defined for every bounded operator $T$ in $\B(H)$ by:
\[ \norm{T} = \sup\set{\norm{T(x)} \mid x \in H, \norm{x} \leq 1}. \]

The norm topology (or uniform topology) is the topology induced by the operator norm on $\B(H)$.

A subbase for this topology is given by the sets $\set{ B \in \B(H) \mid \norm{A - B} < \epsilon}$ where $A \in \B(H)$ and $\epsilon > 0$.

A sequence of bounded operators $(T_n)$ converges to a bounded operator $T$ in this topology if and only if $\norm{T_n-T} \underset{n \rightarrow \infty}{\longrightarrow} 0$.
\end{mydef}

\begin{mydef}
The strong operator topology (or SOT) on $\B(H)$ is the topology of pointwise convergence in the norm of $H$: a net of bounded operators $(T_\lambda)_{\lambda \in \Lambda}$ converges to a bounded operator $T$ in this topology if and only if $\norm{(T_\lambda - T)x} \longrightarrow 0$ for each $x \in H$. In that case, $T$ is said to be strongly continuous (or SOT-continuous).

A subbase for this topology is given by the sets $\set{ B \in \B(H) \mid \norm{(A - B)x} < \epsilon}$ where $x \in H$, $A \in \B(H)$ and $\epsilon > 0$.
\end{mydef}

\begin{mydef}
The weak operator topology (or WOT) on $\B(H)$ is the topology of pointwise weak convergence in the norm of $H$: a net of bounded operators $(T_\lambda)_{\lambda \in \Lambda}$ converges to a bounded operator $T$ in this topology if and only if $\prodscal{(T_\lambda - T)x}{y} \longrightarrow 0$ for $x, y \in H$.  In that case, $T$ is said to be weakly continuous (or WOT-continuous).

A subbase for this topology is given by the sets $\set{ B \in \B(H) \mid \prodscal{(A - B)x}{y} < \epsilon}$ where $x,y \in H$, $A \in \B(H)$ and $\epsilon > 0$.
\end{mydef}

The word "operator" is often omitted.

\begin{myprop}
Let $H$ be a Hilbert space. The weak operator topology on $\B(H)$ is weaker than the strong operator topology on $\B(H)$.
\end{myprop}

\begin{proof}
Let $(T_\lambda)_{\lambda \in \Lambda}$ be a net of bounded operators in $\B(H)$.
Suppose that $(T_\lambda)_{\lambda \in \Lambda}$ converges strongly to a bounded operator $T \in \B(H)$. Then, $\norm{(T_\lambda - T)x} \longrightarrow 0$ for every $x \in H$. 

By the Cauchy-Schwarz inequality, we observe that $\abs{\prodscal{(T_\lambda - T)x}{y}} \leq \norm{(T_\lambda - T)x} \norm{y}$ for every $x, y \in H$. Thus, $\prodscal{(T_\lambda - T)x}{y} \longrightarrow 0$ for every $x, y \in H$ and therefore, $(T_\lambda)_{\lambda \in \Lambda}$ converges weakly to $T$.
\end{proof}

\begin{myprop}
Let $H$ be a Hilbert space. The strong operator topology on $\B(H)$ is weaker than the norm topology on $\B(H)$.
\end{myprop}

\begin{proof}
Let $(T_\lambda)_{\lambda \in \Lambda}$ be a net of bounded operators in $\B(H)$. Suppose that $(T_\lambda)_{\lambda \in \Lambda}$ converges in the norm topology to a bounded operator $T \in \B(H)$. 

Then, $\norm{T_\lambda - T} = \sup\set{\norm{(T_\lambda -T)(x)} \mid x \in H, \norm{x} \leq 1} \longrightarrow 0$ and therefore, for every $x \in H$, $\norm{(T_\lambda - T)x} \longrightarrow 0$. Thus, $(T_\lambda)_{\lambda \in \Lambda}$ converges strongly to $T$.
\end{proof}

It is known that for every finite-dimensional Hilbert space $H$, the weak topology, the strong topology and the norm topology coincide. Moreover, for the strong and the weak operator topologies, the use of nets instead of sequences should not be considered trivial: it is known that, for an arbitrary Hilbert space $H$, the norm topology is first-countable whereas the other topologies are not necessarily first-countable, see \cite[Chapter II.2]{takesaki1} and \cite[I.3.1]{blackadar}. 

\begin{mydef}
Let $H$ be a Hilbert space and $A \subset \B(H)$.

The commutant of $A$ is the set $A'$ of all bounded operators that commutes with those of $A$:
\[ 
  A' = \set{T \in \B(H) \mid \forall S \in A, TS=ST}
\]

The bicommutant of $A$ is the commutant of $A'$. We denote it by $A''$.
\end{mydef}

\begin{mytheo}[von Neumann bicommutant theorem]\label{vn-bicommutant}
 Let $A$ be a unital *-subalgebra of $\B(H)$ for some Hilbert space $H$. The following conditions are equivalent:
\begin{enumerate}
 \item $A=A''$.
 \item $A$ is closed in the weak topology of $\B(H)$.
 \item $A$ is closed in the strong topology of $\B(H)$.
\end{enumerate}
\end{mytheo}

This theorem is a fundamental result in operator theory as it remarkably relates a topological property (being closed in two operator topologies) to an algebraic property (being its own bicommutant).

\begin{mydef}
A W*-algebra (or von Neumann algebra) is a C*-algebra which satifies one (hence all) of the conditions of the von Neumann bicommutant theorem.
\end{mydef}

It follows that the collections of bounded operators on Hilbert spaces are the most trivial examples of W*-algebras.  



\begin{mydef}
Let $A$ and $B$ be two C*-algebras. 

A positive map $\phi : A \rightarrow B$ is normal if every increasing net $(x_\lambda)_{\lambda \in \Lambda}$ in $A^+$ with least upper bound $\lub x_\lambda \in A^+$ is such that the net $(\phi(x_\lambda))_{\lambda \in \Lambda}$ is an increasing net in $B^+$ with least upper bound 

\[
  \lub \phi (x_\lambda) = \phi(\lub x_\lambda).
\]
\end{mydef}

It is important to note that the notion of normal map (defined in \cite[III.2.2.1]{blackadar} or \cite[Theorem 1, pp.57]{dixmier-vn}) relates to the notion of positive Scott-continuous map, although it is in general not the case that $A^+$ and $B^+$ are dcpos when $\phi : A \to B$ is a positive map between C*-algebras (see Example \ref{cont-ex-cstar-dcpo}). This observation paves the way to interesting connections between operator theory and domain theory, that we will study later.

Moreover, by \cite[III.2.2.2]{blackadar}, we know that normal maps and positive weak-continuous maps coincide. Thus, the W*-algebras and the normal sub-unital maps (or NsU-maps) between them give rise to a category $\wPSU$, which turns out to be a subcategory of the category of C*-algebras $\PSU$.

\section{Effect modules and the subdistribution monad}
In this section, we introduce effect algebras, which are structures that have been introduced in mathematical physics to study quantum probability and quantum logic in the same setting \cite{new-trends}. The relation between effect algebras and the distribution monad as already been studied in \cite{manifesto}. We will now investigate the relationship between effect algebras and the subdistribution monad, in order to study non-terminating probabilistic programs.

\subsection{Effect algebras and effect modules}
\begin{mydef}
A partial commutative monoid (PCM) is a set $M$ equipped with a zero element $0 \in M$ and a partial binary operation $\ovee : M \times M \rightarrow M$ satisfying the following properties (where $x \perp y$ is a notation for "$x \ovee y$ is defined")
\begin{description}
 \item[Commutativity] $x \perp y$ implies $y \perp x$ and $x \ovee y = y \ovee x$.
 \item[Associativity] $y \perp z$ and $x \perp (y \ovee z)$ implies $x \perp y$ and $(x \ovee y) \perp z$ and also $x \ovee (y \ovee z) = (x \ovee y) \ovee z$.
 \item[Zero] $0 \perp x$ and $0 \ovee x = x$.
\end{description}
\end{mydef}

When writing $x \ovee y$, we shall now implicitly assume that $x \perp y$.

\begin{mydef}
An effect algebra $(E,0,\ovee, (-)^\perp)$ is a PCM $(E,0,\ovee)$ together with an unary operation $(-)^\perp : E \rightarrow E$ satisfying
\begin{enumerate}
 \item $x^\perp \in E$ is the unique element in $E$ such that $x \perp x^\perp$ and $x \ovee x^\perp = 1$, where $1 = 0^\perp$;
 \item $x \perp 1 \implies x=0$.
\end{enumerate}

A homomorphism of effect algebras is a function $f : E \rightarrow F$ between the underlying sets satisfying $f(1)=1$, and if $x \perp x'$ in $E$, then $f(x) \perp f(x')$ in $F$ and $f(x \ovee x') = f(x) \ovee f(x')$.

We write $\EA$ for the category of effect algebras together with such homomorphisms.
\end{mydef}

\begin{mydef}
A generalized effect algebra is a PCM $(E,0,\ovee)$ satisfying the following properties:
\begin{description}
 \item[Cancellation law] If $x \ovee y = x \ovee z$ then $y=z$.
 \item[Positivity law] If $x \ovee y = 0$ then $x = y = 0$. 
\end{description}

A homomorphism of generalized effect algebras is a function $f : E \to F$ between the underlying sets satisfying $f(0)=0$, and if $x \perp x'$ in $E$, then $f(x) \perp f(x')$ in $F$ and $f(x \ovee x') = f(x) \ovee f(x')$.

We write $\GEA$ for the category of generalized effect algebras together with such homomorphisms.
\end{mydef}

\begin{mydef}\label{def-operation}
Let $(E,0,\ovee, (-)^\perp)$ be an effect algebra.

The dual operation $\owedge$ of the partial sum $\ovee$ is defined by $x \owedge y = (x^\perp \ovee y^\perp)^\perp$ ($x,y \in E$).

The difference operation $\ominus$ is defined by $y \ominus x = z \Leftrightarrow y = x \ovee z$ ($x,y,z \in E$). 

Furthermore, for every effect algebra $E$, one can define a partial order with $1$ as top and $0$ as bottom: $x \leq y$ if and only if $\exists z. x \ovee z = y$ \cite[Lemma 5]{manifesto}. 
\end{mydef}

It was shown in \cite[Section 1.2]{new-trends} that, for every generalized effect algebra $E$, one can define the same partial order $\leq$ with $0$ as bottom, and a top $1$ if and only if $E$ is an effect algebra. In other words, an effect algebra is a generalized effect algebra with a top.


\begin{mylem}\label{homo-ea}
If $f : E \to F$ is a homomorphism of effect algebras, then $f(x^\perp)=f(x)^\perp$ and thus $f(0)=0$.

If $f : E \to F$ is a homomorphism of generalized effect algebras, then $x \leq_E x' \implies f(x) \leq_F f(x')$
\end{mylem}

\begin{proof}
Suppose that $f \in \EA(E,F)$. Let $x$ be an element of $E$. Then, $x \perp x^\perp$ and $x \ovee x^\perp = 1$. Since $f$ preserves the sum and the unit, we obtain that $f(x) \perp f(x^\perp)$ and $f(x) \ovee f(x^\perp) = f(x \ovee x^\perp)=f(1)=1$. It follows that $f(x^\perp)=f(x)^\perp$  by uniqueness of the orthocomplement of $f(x)$. In particular, when $x=0$, $f(0)=f(1^\perp)=f(1)^\perp=1^\perp=0$.

Suppose that $f \in \GEA(E,F)$. Let $x$ and $y$ be two elements of $E$. If $x \leq_E y$, then there is a $z \in E$ such that $x \ovee z = y$. We obtain that $f(x) \ovee f(z) = f(x \ovee z) = f(y)$ where $f(z) \in F$. That is to say $f(x) \leq_F f(y)$.
\end{proof}

It should be noted that homomorphisms of generalized effect algebras do not necessarily preserve the orthocomplement. For example, for the map $f : [0,1]_\R \to [0,1]_\R$ defined by $f(x)=\frac{1}{2}x$, it turns out that $f(x \ovee y)=\frac{1}{2}(x \ovee y)=\frac{1}{2}x \ovee \frac{1}{2}y = f(x) \ovee f(y)$ and $f(0)=0$ but $f(1)^\perp = 1 \ominus f(1) = \frac{1}{2} \neq 0 = f(0) = f(1^\perp)$.

\begin{mydef}
For every effect algebra $E$ and every $t \in E$, we define the downset 
\[
 {\downarrow} t = \set{ x \in E \mid 0 \leq x \leq t}
\]
\end{mydef}

\begin{myprop}
Let $(E,\ovee,0,(-)^\perp)$ be an effect algebra and $t \in E$.

The downset ${\downarrow} t$ is an effect algebra with the sum $\ovee$ restricted to ${\downarrow} t$, the element $t$ as top, the orthocomplement defined by $x^\perp = t \ominus x$ for every element $x \in {\downarrow} t$, and finally $x \perp y$ if and only if $x \perp_E y$ and $x \ovee y \leq t$ ($x,y \in {\downarrow} t$).
\end{myprop}

\begin{proof}
Let $x, y, z \in {\downarrow} t$.

\emph{Commutativity}: If $x \perp y$, then $x \perp_E y$ and $x \ovee y \leq t$. It follows that $y \perp_E x$ and $y \ovee x = x \ovee y \leq t$. That is to say $y \perp x$.

\emph{Associativity}: Suppose that $y \perp z$ and $x \perp (y \ovee z)$. Then $y \perp_E z$ and $y \ovee z \leq t$ and $x \perp_E (y \ovee z)$ and also $x \ovee (y \ovee z) \leq t$. Thus, $x \perp_E y$, $(x \ovee y) \perp_E z$ and $x \ovee y \leq (x \ovee y) \ovee z = x \ovee (y \ovee z) \leq t$. It follows that $x \perp y$ and $(x \ovee y) \perp z$.

\emph{Zero}: From $x \in {\downarrow} t \subseteq E$, we obtain that $0 \perp_E x$ and $0 \ovee x = x \leq t$. That is to say $0 \perp x$.

Hence, ${\downarrow} t$ is a PCM. We now consider an orthocomplement for ${\downarrow} t$: \\
By \cite[Lemma 6(vii)]{manifesto}, $x \leq t$ implies $t \ominus (t \ominus x) = x$ and thus, by Definition \ref{def-operation}, $x \ovee (t \ominus x) = t$. We obtain that $x^\perp = t \ominus x$ by unicity of the orthocomplement. Moreover, $x \perp t$ implies that $x \perp_E t$ and $x \ovee t \leq t$, and thus $x \leq t \ominus t = 0$ by \cite[Lemma 6(iv)]{manifesto}. Then, $x \leq 0$, which implies that $x = 0$.
\end{proof}

\begin{myprop}
Let $f : E \to F$ be a function between two effect algebras $E$ and $F$.

Let $\tilde{f} : E \to f(E)= {\downarrow} f(1) := \set{ x \in F \mid 0 \leq x \leq f(1)}$ be the function defined pointwise by $\tilde{f}(x)=f(x)$.

Then $f$ is a homomorphism of $\GEA$ (i.e. $f \in \GEA(E,F)$) if and only if $\tilde{f}$ is a homomorphism of $\EA$ (i.e. $\tilde{f} \in \EA(E,{\downarrow} f(1))$)
\end{myprop}

\begin{proof}
Let $x, y \in E$.
 
Suppose that $f \in \GEA(E,F)$ and that $x \perp_E y$. Then $f$ preserves the sum and thus $\tilde{f}(x) = f(x) \perp_F f(y) = \tilde{f}(y)$. Since $x \ovee y \leq 1$ implies that $\tilde{f}(x) \ovee \tilde{f}(y) = f(x) \ovee f(y) = f(x \ovee y) \leq f(1)$, we have that $\tilde{f}(x) \perp_{{\downarrow} f(1)} \tilde{f}(y)$ and $\tilde{f}(x \ovee y) = f(x \ovee y) = f(x) \ovee f(y) = \tilde{f}(x) \ovee \tilde{f}(y)$. Hence, $\tilde{f}$ preserves the sum.


Moreover, the unit of ${\downarrow} f(1)$ is $f(1)=\tilde{f}(1)$. That is to say, $\tilde{f}$ also preserves the unit. It follows that $\tilde{f}$ is a homomorphism of $\EA$.

Conversely, suppose that $\tilde{f} \in \EA(E,{\downarrow} f(1))$ and that $x \perp_E y$. $\tilde{f}$ preserves the sum and thus $f(x) = \tilde{f}(x) \perp_{{\downarrow} f(1)} \tilde{f}(y) = f(y)$, which is equivalent to say that $f(x) \perp_F f(y)$ and $f(x) \ovee f(y) \leq f(1)$. Since $f(x) \ovee f(y) = \tilde{f}(x) \ovee \tilde{f}(y) = \tilde{f}(x \ovee y) = f (x \ovee y)$, we conclude that $f$ preserves the sum.

Moreover, $\tilde{f}$ preserves zero (by Lemma \ref{homo-ea}) and therefore, $f$ preserves zero since $f(0)=\tilde{f}(0)=0$. 
Hence, $f$ is a homomorphism of $\GEA$.
\end{proof}


We will now introduce effect modules, which are the effect-theoretic counterpart of vector spaces.

\begin{mydef}
A (generalized) effect module is a (generalized) effect algebra $E$ together with a scalar multiplication $r \bullet x \in E$, where $x \in E$ and $r \in [0,1]$, satisfying :
\begin{align*}
 1 \bullet x  = x && (r+s) \bullet x = r \bullet x + s \bullet x &\quad \text{ if } r+s \leq 1 \\
 (rs) \bullet x = r \bullet (s \bullet x) && r \bullet (x \ovee y) = (r \bullet x) \ovee (r \bullet y) &\quad \text{ if } x \perp y
\end{align*}

A map of (generalized) effect modules is a map of (generalized) effect algebras $f : E \to F$ which preserves scalar multiplication, i.e. $f(r \bullet x) = r \bullet f(x)$ with $x \in E$ and $r \in [0,1]$.

We write $\EMod$ for the category of effect modules with homomorphisms of effect modules. Similarly, we write $\GEMod$  for the category of generalized effect modules with homomorphisms of generalized effect modules.
\end{mydef}

It was observed in \cite{furber-jacobs} that for every C*-algebra $A$, the subset of effects $[0,1]_A$ is an effect algebra $([0,1]_A,0,+)$ with $x \perp y$ if and only if $x + y \leq 1$ and the orthocomplement $x^\perp = 1 - x$. It is therefore an effect module with a $[0,1] \subseteq \R$ scalar multiplication where $r \bullet x \in [0,1]_A$ for $r \in [0,1]$ and $x \in [0,1]_A$.

\subsection{Discrete probability monads}
We now introduce the distribution monad and the subdistribution monad, which are heavily used for studying discrete probability systems such as Markov chains, see \cite{jacobs-coalgebra}.

\begin{mydef}
The distribution monad $\DM : \mathbf{Sets} \to \mathbf{Sets}$ is the monad defined by 
\[
  \DM(X) = \{\phi: X \rightarrow [0,1] \mid \sum_{x \in X} \phi(x) = 1\}
\]
\end{mydef}

From now, we will use "$\lambda x. \cdots$" as a notation for "$x \mapsto \cdots$".

\begin{mydef}[\cite{hasuo-jacobs-sokolova}]
The subdistribution monad $\SDM : \mathbf{Sets} \rightarrow \mathbf{Sets}$ is the monad defined by \[\SDM(X) = \{ \phi : X \rightarrow [0,1] \mid \sum_{x \in X} \phi(x) \leq 1\}.\]
Its unit $\eta : X \rightarrow \SDM(X)$ and multiplication $\mu : \SDM^2(X) \rightarrow \SDM(X)$ are given by

\begin{alignat*}{1}
   \eta(x)=\lambda x'. \left\{\begin{matrix} 0 & \mbox{if } x \ne x' \\1 & \mbox{if } x=x'   \end{matrix}\right. &\qquad \mu(\Phi)(x)=\sum_\phi \Phi(\phi) \cdot \phi(x).
\end{alignat*}
\end{mydef}

The possibility of a missing probability in the subdistribution monad can be seen as a probability of deadlock. Any morphism $\phi : X \rightarrow [0,1]$ such that $\sum_{x \in X} \phi(x) \leq 1$ can be seen as a "deadlock-sensitive" morphism $\varphi : 1+X \rightarrow [0,1]$ defined by 
\[
  \varphi(x)=\left\{\begin{matrix*}[l] \phi(x) & \mbox{if } x \in X \\ 1-\sum_{x \in X} \phi(x) & \mbox{if } x \notin X \end{matrix*}\right.
\]

Hence, one can write for any set $X$ that $\SDM(X)$ is isomorphic to $\DM(1+X)$. In the same manner, Prakash Panangaden used in \cite{panangaden} a similar "subprobability measure" and stated that a probability measure on $1+X$ is a subprobability measure on $X$.

Moreover, for every set $X$, we now identify every element $\phi \in \SDM(X)$ with a subconvex sum $\sum_i r_i x_i$ with $x_i \in X$ and $r_i \in [0,1]$ such that $\sum_i r_i \leq 1$. A subconvex set is an Eilenberg-Moore algebra of the subdistribution monad $\SDM$. Namely, a subconvex set consists of a set $X$ in which the subconvex sums $\sum_i r_i x_i \in X$ exists for all subconvex combinations. We now define the category $\SConv = \EiMo(\SDM)$ with subconvex sets as objects and affine maps preserving convex sums as morphisms.

The interested reader will find in Appendix \ref{appendix-sub-adj} a proof of the following adjunction between $\SConv$ and $\GEMod$, inspired by \cite{jacobs-conv-effects,jorik-jacobs}:

$$\xymatrix@R+.5pc{
\opp{\GEMod}\ar@/^1.5ex/[rr]^-{\GEMod(-,[0,1])} 
   & \top & \SConv\ar@/^1.5ex/[ll]^-{\SConv(-,[0,1])} 
}$$

\section{Quantum domain theory: Definitions}
Domain theory was successfully applied by Jones and Plotkin \cite{jones-plotkin} in the context of probabilistic computation. In this section, we investigate the relevance of W*-algebras in the extension of domain theory to the quantum setting, following the work of Selinger \cite{selinger}. 

\subsection{W*-algebras as directed-complete partial orders}
Since positive elements are self-adjoint, one can define the following order on positive maps of C*-algebras.

\begin{mydef}[Löwner partial order]
For positive maps $f,g : A \rightarrow B$ between C*-algebras $A$ and $B$, we define pointwise the following partial order $\lwnr$, which turns out to be an infinite-dimensional generalization  of the Löwner partial order \cite{lowner} for positive maps: $f \lwnr g$ if and only if $\forall x \in A^+, f(x) \leq g(x)$ if and only if $\forall x \in A^+, (g-f)(x) \in B^+$ (i.e. $g-f$ is positive).
\end{mydef}

One might ask if, for arbitrary C*-algebras $A$ and $B$, the poset $(\PSU(A,B),\lwnr)$ is directed-complete. The answer turns out to be no, as shown by our following counter-example.

\begin{mysam}\label{cont-ex-cstar-dcpo}
Let us consider the C*-algebra $C([0,1]):=\set{f: [0,1] \rightarrow \C \mid f \text{ continuous}}$.

The hom-set $\PSU(\C,C([0,1]))$ is isomorphic to $C([0,1])$ if one considers the functions $F : \PSU(\C,C([0,1])) \to C([0,1])$ and $G : C([0,1]) \to \PSU(\C,C([0,1]))$ respectively defined by $F(f)=f(1)$ and $G(g)=\lambda \alpha \in \C. \alpha \cdot g$.

We define an increasing chain $(f_n)_{n \geq 0}$ of $C([0,1])$ define for every $n \in \N$ by 

\[ f_n(x) = \left\{
  \begin{array}{l l}
    0 & \quad \text{if } 0 \leq x < \frac{1}{2}\\
    (x-\frac{1}{2})2^{n+1} & \quad \text{if } \frac{1}{2} \leq x \leq \frac{1}{2} + 2^{-(n+1)}\\
    1 & \quad \text{if } \frac{1}{2} + 2^{-(n+1)} < x \leq 1
  \end{array} \right.\] 
  
Suppose that there is a least upper bound $\phi$ in $C([0,1])$ for this chain. Then, $\phi(x)=0$ if $x < \frac{1}{2}$. Moreover, $\lim_{n \to \infty} \left( \frac{1}{2} + 2^{-(n+1)} \right) = \frac{1}{2}$ implies that $\phi(x)=1$ if $x > \frac{1}{2}$. It follows that $\phi(x) \in \set{0,1}$ if $x \neq \frac{1}{2}$.

By the Intermediate Value Theorem, the continuity of the function $\phi$ on the interval $[0,1]$ implies that there is a $c \in [0,1]$ such that $\phi(c)=\frac{1}{2}$. From $\phi(c) \notin \set{0,1}$, we obtain that $c = \frac{1}{2}$. That is to say $\phi(\frac{1}{2})=\frac{1}{2}$, which is absurd since $f_n(\frac{1}{2})=0$ for every $n\in \N$.



It follows that there is no least upper bound for this chain in $C([0,1])$ and therefore $C([0,1])$ is not chain-complete.
\end{mysam}

\begin{mytheo}\label{wstar-dcpo}
For W*-algebras $A$ and $B$, the poset $(\wPSU(A,B),\lwnr)$ is directed-complete.
\end{mytheo}

The proof of this theorem will be postponed until after the following lemmas.

\begin{mylem}\label{pointwise-lub}
Let $f : A \to B$ be a PsU-map between unital C*-algebras $A$ and $B$ and $x \in A^+$. 

Then, $f(x) \leq \norm{x} \scprod 1$.
Therefore, $\norm{f(x)} \leq \norm{x}$. 
\end{mylem}

\begin{proof}
Let $f : A \to B$ be a PsU-map between unital C*-algebras $A$ and $B$ and $x \in A^+$. 

Then, $x \leq \norm{x} \scprod 1$ by \cite[Proposition 4.2.3(ii)]{kadison-ringrose1}. Thus, $f(x) \leq \norm{x} \scprod f(1) \leq \norm{x} \scprod 1$ since $f(1) \leq 1$. Hence, $\norm{f(x)} \leq \norm{x}$. 
\end{proof}

The following result is known in physics as Vigier's theorem \cite{vigier}. A weaker version of this theorem can be found in \cite{selinger}. It is important in this context because it establishes the link between limits in topology and joins in order theory.

\begin{mylem}\label{strong-lub}
Let $H$ be a Hilbert space. Let $(T_\lambda)_{\lambda \in \Lambda}$ be an increasing net of $\Ef(H)$. \\
Then the least upper bound $\bigvee T_\lambda$ exists in $\Ef(H)$ and is the limit of the net $(T_\lambda)_{\lambda \in \Lambda}$ in the strong topology.
\end{mylem}

\begin{proof}
For any operator $U \in \B(H)$, the inner product $\prodscal{U x}{x}$ is real if and only if $U$ is self-adjoint (by Theorem \ref{self-adjoint-BH}). Thus, for each $x \in H$, the net $(\prodscal{T_\lambda x}{x})_{\lambda \in \Lambda}$ of real numbers is increasing, bounded by $\norm{x}^2$ and thus convergent to a limit $\lim_\lambda \prodscal{T_\lambda x}{x}$ since $\R$ is bounded-complete. 

By polarization on norms, $\prodscal{T_\lambda x}{y} = \frac{1}{2} (\prodscal{T_\lambda (x+y)}{(x+y)} - \prodscal{T_\lambda x}{x} - \prodscal{T_\lambda y}{y})$ for any $\lambda \in \Lambda$. 

Then, for all $x, y \in H$, the limit $\lim_\lambda \prodscal{T_\lambda x}{y}$ exists and thus we can define pointwise an operator $T \in \Ef(H)$ by $\prodscal{T x}{y} = \lim_\lambda \prodscal{T_\lambda x}{y}$ for $x,y \in H$. 

Indeed, $T$ is the limit of the net $(T_\lambda)_{\lambda \in \Lambda}$ in the weak topology, and therefore in the strong topology since a bounded net of positive operators converges strongly whenever it converges weakly (see \cite[I.3.2.8]{blackadar}). 

Moreover, $T$ is an upper bound for $(T_\lambda)_{\lambda \in \Lambda}$ since $T_\lambda \leq T$ for every $\lambda \in \Lambda$. By Theorem \ref{positive-BH}, if there is a self-adjoint operator $S \in B(H)$ such that $T_\lambda \leq S$ for every $\lambda \in \Lambda$, then  $\prodscal{T_\lambda x}{x} \leq \prodscal{S x}{x}$ for every $\lambda \in \Lambda$. Thus, $\prodscal{T x}{x} = \lim_\lambda \prodscal{T_\lambda x}{x} \leq \prodscal{S x}{x}$. Then, $\prodscal{(S-T)x}{x} \geq 0$ for every $x \in H$. By Theorem \ref{positive-BH}, $S-T$ positive and thus $T \leq S$. It follows that $T$ is the least upper bound of $(T_\lambda)_{\lambda \in \Lambda}$.
\end{proof}

\begin{mycoro}\label{effects-ccpo}
For every W*-algebra $A$, the poset $[0,1]_A$ is directed-complete.
\end{mycoro}

\begin{proof}
Let $A$ be a W*-algebra. By definition, $A$ is a strongly closed subalgebra of $\B(H)$, for some Hilbert space $H$.

Let $(T_\lambda)_{\lambda \in \Lambda}$ be an increasing net in $[0,1]_A \subseteq \Ef(H)$. By Lemma \ref{strong-lub}, $(T_\lambda)_{\lambda \in \Lambda}$ converges strongly to $\lub T_\lambda \in \Ef(H)$. It follows that $\lub T_\lambda \in [0,1]_A$ because $[0,1]_A$ is strongly closed. Thus, $[0,1]_A$ is directed-complete.
\end{proof}

This corollary constitute a crucial step in the proof of Theorem \ref{wstar-dcpo}, as it unveils a link between the topological properties and the order-theoretic properties of W*-algebras.

\begin{mylem}\label{psu-determined}
Any positive map $f : A \to B$ between C*-algebras is completely determined and defined by its action on $[0,1]_A$. Hence, the functor $[0,1]_{(-)} : \PSU \to \GEMod$ is full and faithful.
\end{mylem}

\begin{proof}
A positive map of C*-algebras $f : A \to B$ restrict by definition to a map $f : A^+ \to B^+$. By Lemma \ref{presrv-order}, $f$ preserves the order $\leq$ on positive elements and thus restricts to $[0,1]_A \to [0,1]_B$:

Let $x \in A^+ \setminus \set{0}$. From $x \leq \norm{x} 1$, we can see that $\frac{1}{\norm{x}}x \in [0,1]_A$ and thus $f(\frac{1}{\norm{x}}x) \in [0,1]_B$. Moreover, $f(x)=\norm{x}f(\frac{1}{\norm{x}}x)$. This statement can be extended to every element in $A$ since each $y \in A$ is a linear combination of four positive elements (see \cite[II.3.1.2]{blackadar}), determining $f(y) \in B$.
\end{proof}

\begin{proof}[Proof of Theorem \ref{wstar-dcpo}]
Let $A$ and $B$ be two W*-algebras. By Corollary \ref{effects-ccpo}, $[0,1]_A$ and $[0,1]_B$ are directed-complete. 


We now consider an increasing net $(f_\lambda)_{\lambda \in \Lambda}$ of NsU-maps from $A$ to $B$, increasing in the Löwner order.
Then, for every $x \in A^+$, there is an increasing net $(f_\lambda(x))_{\lambda \in \Lambda}$ bounded by $\norm{x} \scprod 1$ (by Lemma \ref{pointwise-lub}). 

Moreover, for every non-zero element $x \in A^+$, from the fact that $[0,1]_B$ is directed-complete, we obtain that the increasing net $(f_\lambda(\frac{x}{\norm{x}}))_{\lambda \in \Lambda}$ has a least upper bound $\bigvee f_\lambda(\frac{x}{\norm{x}})$ in $[0,1]_B$ and thus we can define pointwise the following upper bound $f : [0,1]_A \to [0,1]_B$ for the increasing net $(f'_\lambda)_{\lambda \in \Lambda}$ of NsU-maps from $[0,1]_A$ to $[0,1]_B$ such that, for every $\lambda \in \Lambda$, $f'_\lambda(x) = f_\lambda(\frac{x}{\norm{x}}) (x \neq 0)$:

\[
  f(\frac{x}{\norm{x}}) = \bigvee f_\lambda(\frac{x}{\norm{x}}) \quad (x \in A^+ \setminus \set{0})
\]

This upper bound $f$ is a positive sub-unital map by construction and can be extended to an upper bound $f : A \to B$ for the increasing net $(f_\lambda)_{\lambda \in \Lambda}$: for every nonzero $x \in A^+$, the increasing sequence $(f_\lambda(x))_{\lambda \in \Lambda}=(\norm{x} f_\lambda(\frac{x}{\norm{x}}))_{\lambda \in \Lambda}$ has a least upper bound $\bigvee f_\lambda (x) = \norm{x} \bigvee f_\lambda(\frac{x}{\norm{x}})$ in $B^+$ and thus one can define pointwise an upper bound $f : A \to B$ for $(f_\lambda)_{\lambda \in \Lambda}$ by $f(x)=\bigvee f_\lambda(x)$ for every $x \in A^+$.

We now need to prove that the map $f$ is normal, by exchange of joins. \\
Let $(x_\gamma)_{\gamma \in \Gamma}$ be an increasing bounded net in $A^+$ with least upper bound $\lubgamma x_\gamma$. 
For every $\gamma' \in \Gamma$, we observe that $x_{\gamma'} \leq \lubgamma x_\gamma$ and thus, by Lemma \ref{presrv-order}, $f(x_{\gamma‘}) \leq f(\lubgamma x_\gamma)$. As seen earlier, since $[0,1]_B$ is directed-complete, the increasing net $(f(x_\gamma))_{\gamma \in \Gamma}$, which is equal by definition to the increasing net $(\lublambda (f_\lambda (x_\gamma)))_{\gamma \in \Gamma}$, has a least upper bound in $B^+$ defined by $\lubgamma f(x_\gamma)=\lub_{\gamma \in \Gamma, x_\gamma \neq 0} \norm{x_\gamma} f(\frac{1}{\norm{x_\gamma}}x_\gamma)$ if there is a $\gamma'' \in \Gamma$ such that $x_{\gamma''} \neq 0$ and by $\lubgamma f(x_\gamma) = 0$ otherwise. It follows that $\lubgamma f(x_\gamma) \leq f(\lubgamma x_\gamma)$.

We have to prove now that $f(\lubgamma x_\gamma) \leq \lubgamma f(x_\gamma)$. Since each map $f_\lambda$ ($\lambda \in \Lambda$) is normal, we obtain that $f(\lubgamma x_\gamma)=\lublambda(f_\lambda(\lubgamma x_\gamma))=\lublambda(\lubgamma(f_\lambda(x_\gamma))$. Moreover, for $\gamma' \in \Gamma$ and $\lambda' \in \Lambda$, $f_{\lambda'}(x_{\gamma'}) \leq \lublambda f_\lambda(x_{\gamma'}) \leq \lubgamma (\lublambda f_\lambda(x_{\gamma}))$. Then, $\lubgamma f_{\lambda'}(x_{\gamma}) \leq \lubgamma (\lublambda f_\lambda(x_{\gamma}))$ and thus $\lublambda (\lubgamma f_{\lambda}(x_{\gamma})) \leq \lubgamma (\lublambda f_\lambda(x_{\gamma}))$. It follows that $f(\lubgamma x_\gamma) = \lublambda (\lubgamma f_{\lambda}(x_{\gamma})) \leq \lubgamma (\lublambda f_\lambda(x_{\gamma})) = \lubgamma f(x_\gamma)$.

Let $g \in \wPSU(A,B)$ be an upper bound for the increasing net $(f_\lambda)_{\lambda \in \Lambda}$. For $\lambda' \in \Lambda$ and $x \in A^+$, $f_{\lambda'}(x) \leq g(x)$. Then, $\forall x \in A^+, f(x) = \bigvee f_\lambda(x) \leq g(x)$, i.e. $f \lwnr g$. It follows that $f$ is the least upper bound of $(f_\lambda)_{\lambda \in \Lambda}$.
\end{proof}

\begin{mytheo}
The category $\wPSU$ is a $\Dcpo_\perp$-enriched category.
\end{mytheo}

\begin{proof}
The category $\Dcpo_\perp$ is cartesian closed and therefore monoidal.

For every pair $(A,B)$ of W*-algebras, $\wPSU(A,B)$ together with the Löwner order is a dcpo with  zero map as bottom, and therefore $\wPSU(A,B) \in \Dcpo_\perp$.

In particular, for every W*-algebra $A$, $\wPSU(A,A) \in \Dcpo_\perp$. The element $1:=\set{\bot}$ is the terminal object of the cartesian closed category $\Dcpo_\perp$. We consider now for every W*-algebra $A$ a map $I_A : 1 \to \wPSU(A,A)$ such that $I_A(\bot) \in \wPSU(A,A)$ is the identity map on $A$. The map $I_A$ is clearly Scott-continuous for every W*-algebra $A$.

Then, what need to be proved is that, given three W*-algebras $A, B, C$, the composition $\circ_{A,B,C} : \wPSU(B,C) \times \wPSU(A,B) \to \wPSU(A,C)$ is Scott-continuous. By \cite[Lemma 3.2.6]{abramsky-jung}, it is equivalent to show that:
\begin{itemize}
 \item for every NsU-map $f : A \to B$, the precomposition $(-) \circ f : \wPSU(B,C) \to \wPSU(A,C)$ given by $g \mapsto g \circ f$ is Scott-continuous.
 \item for every NsU-map $h : B \to C$, the postcomposition $h \circ (-) : \wPSU(A,B) \to \wPSU(A,C)$ given by $g \mapsto h \circ g$ is Scott-continuous.
\end{itemize} 

We now consider a NsU-map $f : A \to B$ and the increasing net $(g_\lambda)_{\lambda \in \Lambda}$ in $\wPSU(B,C)$, with least upper bound $\lublambda g_\lambda \in \wPSU(B,C)$. One can define an upper bound pointwise by $u(x) = ((\lublambda g_\lambda) \circ f)(x)$ for the increasing net $(g_\lambda \circ f)_{\lambda \in \Lambda}$ in $\wPSU(A,C)$. It is easy to check that $u$ is a least upper bound for the increasing net $(g_\lambda \circ f)_{\lambda \in \Lambda}$: for every upper bound $v \in \wPSU(A,C)$ of the increasing net $(g_\lambda \circ f)_{\lambda \in \Lambda}$, we have that $\forall \lambda \in \Lambda, g_\lambda \circ f \lwnr v$, i.e. $\forall \lambda \in \Lambda, \forall x \in A^+, g_\lambda (f(x)) \leq v(x)$ and thus $\forall x \in A^+, u(x)=((\lublambda g_\lambda) \circ f)(x) = (\lublambda g_\lambda)(f(x)) \leq v(x)$, which implies that $u \lwnr v$. It follows that the precomposition is Scott-continuous and, similarly, the postcomposition is Scott-continuous.
\end{proof}

\subsection{Monotone-complete C*-algebras}
A posteriori, we found out that it is known that the subset of effects of an arbitrary W*-algebra is directed-complete \cite[III.3.13-16]{takesaki1} but it is probably the first time that one formulates, proves and strengthens this fact from a domain-theoretic point of view. Moreover, it turns out that Theorem \ref{wstar-dcpo} can be slightly generalize to the following theorem.

\begin{mytheo}\label{cstar-dcpo}
Let $A$ and $B$ be two C*-algebras. \\
If $[0,1]_B$ is directed-complete, then the poset $(\PSU(A,B),\lwnr)$ is directed-complete.
\end{mytheo}

\begin{proof}
Let $(f_\lambda)_{\lambda \in \Lambda}$ be an increasing net of PsU-maps from a C*-algebra $A$ to a C*-algebra $B$.

It follows that for every $x \in A^+$, there is an increasing net $(f_\lambda(x))_{\lambda \in \Lambda}$ in $B^+$ bounded by $\norm{x} \scprod 1$ (by Lemma \ref{pointwise-lub}).

We now assume that $[0,1]_B$ is directed-complete. Then, as in the proof of Theorem \ref{wstar-dcpo}, for every nonzero $x \in A^+$, the increasing net $(f_\lambda(x))_{\lambda \in \Lambda}$, which is equal to $(\norm{x} f_\lambda(\frac{x}{\norm{x}}))_{\lambda \in \Lambda}$, has a least upper bound $\bigvee f_\lambda (x) = \norm{x} \bigvee f_\lambda(\frac{x}{\norm{x}})$ in $B^+$.

Thus, we can define for $(f_\lambda)_{\lambda \in \Lambda}$ an upper bound $f : x \mapsto \bigvee f_\lambda (x)$ , which is positive and sub-unital by construction.

Let $g \in \PSU(A,B)$ be an upper bound for the increasing net $(f_\lambda)_{\lambda \in \Lambda}$. For $\lambda' \in \Lambda$ and $x \in A^+$, $f_{\lambda'}(x) \leq g(x)$, which implies that $\forall x \in A^+, f(x) = \bigvee f_\lambda (x) \leq g(x)$, i.e. $f \lwnr g$. Thus, $f$ is the least upper bound of $(f_\lambda)_{\lambda \in \Lambda}$.
\end{proof} 

In operator theory, a C*-algebra is monotone-complete (or monotone-closed) if it is directed-complete for bounded increasing nets of positive elements. The notion of mo\-no\-to\-ne-com\-ple\-te\-ness goes back at least to Dixmier \cite{dixmier-stone} and Kadison \cite{kadison-wdcpo} but, to our knowledge, it is the first time that the notion of mo\-no\-to\-ne-com\-ple\-te\-ness is explicitly related to the notion of directed-completeness.  The interested reader will find in Appendix \ref{appendix-correspondence} a more detailed correspondence between operator theory and order theory.

It is natural to ask if all monotone-complete C*-algebras are W*-algebras. Dixmier proved that every W*-algebra is a monotone-complete C*-algebra and that the converse is not true \cite{dixmier-stone}. For an example of a subclass of monotone-complete C*-algebras which are not W*-algebras, we refer the reader to a recent work by Saitô and Wright \cite{saito-wright}. 


\section{A quantum weakest pre-condition calculus}\label{WP}
Dijkstra invented in \cite{dijkstra} the weakest pre-condition calculus, a systematic method to analyse the properties of programs. In this calculus, every program is associated to a statement $s$ and an operation $\wkp(s)$ transforms a proposition $Q$ in a proposition $P := \wkp(s)(Q)$, with the guarantee that $Q$ holds after the execution of the program denoted by $s$ if $P$ holds before the execution. In this context, $P$ is called weakest pre-condition and $Q$ is called post-condition.

More formally, a program is interpreted as a function $s$ from a state space $X$ to a state space $Y$ and is in one-to-one correspondance with a map $\wkp(s)$ which computes the weakest precondition $\wkp(s)(Q)$ from a given post-condition $Q$.

In \cite{manifesto}, Jacobs provided a categorical interpretation of the weakest pre-condition calculus, establishing the following commutative diagram for discrete probabilistic computations, denoted via the distribution monad $\DM$:

$$\xymatrix@R+.5pc{
\opp{\EMod}\ar@/^1.5ex/[rr]^-{\EMod(-,[0,1])} 
   & \top & \Conv\ar@/^1.5ex/[ll]^-{\Conv(-,[0,1])} \\
& \kl{\DM} \ar[ul]^{\mbox{[predicates/effects]}\quad\quad}\ar[ur]_{\quad\mbox{[states]}}
}$$

A similar "state-and-effect triangle" can be provided for discrete (non-terminating) probabilistic computations via the subdistribution monad $\SDM$, see Appendix \ref{appendix-sub-adj}:

$$\xymatrix@R+.5pc{
\opp{\GEMod}\ar@/^1.5ex/[rr]^-{\GEMod(-,[0,1])} 
   & \top & \SConv\ar@/^1.5ex/[ll]^-{\SConv(-,[0,1])} \\
& \kl{\SDM} \ar[ul]^{\mbox{[predicates/effects]} \qquad \PredS\quad\quad}\ar[ur]_{\quad\mbox{[states]}}
}$$

Thus, we can formulate a weakest precondition calculus for discrete subprobabilistic computations, in terms of the following bijective correspondences:

$$
\begin{prooftree}
{\xymatrix{ 
\text{Kleisli maps } X\ar[r]^-{f} & \SDM(Y)}}
\Justifies
\begin{prooftree}
{\xymatrix{ \text{algebra maps }\SDM(X)\ar[r] & \SDM(Y)}}
\Justifies
{\xymatrix{ \text{generalized effect module maps } \PredS(Y)\ar[r]_-{\wkp(f)} & \PredS(X)}}
\end{prooftree}
\end{prooftree}
$$

We found out that a similar result exists for quantum computations, denoted via W*-algebras:

\begin{itemize}
 \item A (generalized) effect module will be called directed-complete if it is directed-complete as a poset. We will now consider directed-complete generalized effect modules with a separating set\footnote{A set of functions $F$ from a set $X$ to a set $Y$ separates the points of $X$ if for every pair of distinct elements $(x,y) \in X \times X$, there exists a function $f \in F$ such that $f(x) \neq f(y)$.} of Scott-continuous states, i.e. Scott-continuous maps from a generalized effect module $X$ to the interval $\unit$. Together with Scott-continuous maps of generalized effect modules, it gives rise to a category $\sdcGEMod$.
 \item The full and faithful functor $[0,1]_{(-)}: \wPSU \to \sdcGEMod$ of Proposition \ref{pred-wstar}, used as $[0,1]_{(-)}: \opp{(\wPSU)} \to \opp{\sdcGEMod}$, will be our "predicate functor" and describes categorically a quantum "logic of effects". 
 \item We now consider a "normal state functor" $\NS : \opp{(\wPSU)} \to \SConv$ defined by:
 \begin{align*}
   \NS(A) & = \wPSU(A,\C) \simeq \sdcGEMod(\unit_A,\unit_\C) \\
   \NS(A \overset{f}{\to} B) & = (-) \circ f : \NS(B) \to \NS(A)
 \end{align*}
 \item There is an adjunction between $\sdcGEMod$ and $\SConv$, by homming into $\unit$.
\end{itemize}

Thus, one obtain the following theorem, which provides a categorical representation of the duality between states and effects via W*-algebras. The detailed proof can be found in Appendix \ref{appendix-quantum-tri}.

\begin{mytheo}\label{states-effects-wpsu}
The following state-and-effect triangle is a commutative diagram: 
$$\xymatrix@R+.5pc{
\opp{\sdcGEMod}\ar@/^1.5ex/[rr]^-{\sdcGEMod(-,[0,1])} 
   & \top & \SConv\ar@/^1.5ex/[ll]^-{\SConv(-,[0,1])} \\
& \opp{(\wPSU)} \ar[ul]^-{[0,1]_{(-)}}\ar[ur]_{\quad\NS}
}$$
\end{mytheo}

The weakest precondition operator $\operatorname{wp}(f) : \unit_B \to \unit_A$ corresponding to a NsU-map $f : B \to A$ between W*-algebras is given by its restriction $f : [0,1]_B \to [0,1]_A$. It follows from Theorem \ref{states-effects-wpsu} that one can define a weakest pre-condition calculus, which involves the following bijective correspondences:

$$
\begin{prooftree}
{\xymatrix{ 
\text{maps } B\ar[r]^-{f} & A \text{ in } \opp{(\wPSU)}}}
\Justifies
\begin{prooftree}
{\xymatrix{ \text{affine maps }\NS(A)\ar[r] & \NS(B) \text{ in } \SConv}}
\Justifies
{\xymatrix{ \text{generalized effect module maps } [0,1]_B\ar[r]_-{\wkp(f)} & [0,1]_A \text{ in } \opp{\sdcGEMod}}}
\end{prooftree}
\end{prooftree}
$$

\section*{Acknowledgments}
The author would like to thank Bart Jacobs, Robert Furber, Bas Spitters, Bas Westerbaan, Jorik Mandemaker and Prakash Panangaden for helpful discussions. 

\newpage

\begin{appendices}
\section{Correspondence between operator theory and order theory}\label{appendix-correspondence}
In this section, we will provide the following correspondence table between operator theory and order theory, where $A$ and $B$ are C*-algebras.

\begin{center}
\begin{tabular}{|c|c|c|}
    \hline
    \textbf{Operator Theory} & \textbf{Order theory} & \textbf{Reference} \\
    \hline
    $A$ monotone-closed & $[0,1]_A$ directed-complete & \ref{mntn-clsd-alt} \\
    \hline
    $f : A \to B$ NsU-map & $f : [0,1]_A \to [0,1]_B$ Scott-continuous PsU-map & \ref{nrml-alt} \\
    \hline
    $A$ W*-algebra & $[0,1]_A$ dcpo with a separating set of normal states & \ref{thm-takesaki}\\
    \hline
\end{tabular}
\end{center}

In the litterature \cite{blackadar,dixmier-vn,takesaki1}, monotone-closed C*-algebras and normal maps are defined as follow.

\begin{mydef}
A C*-algebra $A$ is monotone-closed (or monotone-complete) if every bounded increasing net of positive elements of $A$ has a least upper bound in $A^+$.

A positive map $\phi : A \rightarrow B$ between C*-algebras is normal if every increasing net $(x_\lambda)_{\lambda \in \Lambda}$ in $A^+$ with a least upper bound $\lub x_\lambda \in A^+$ is such that the net $(\phi(x_\lambda))_{\lambda \in \Lambda}$ is an increasing net in $B^+$ with least upper bound $\lub \phi (x_\lambda) = \phi(\lub x_\lambda)$.
\end{mydef}

In the standard definition of the notion of monotone-closedness, the increasing nets are not required to be bounded by the unit, like in the definitions we used in this thesis. We will now show that we can assume that the upper bound is the unit, without loss of generality.

\begin{myprop}\label{mntn-clsd-alt}
A C*-algebra $A$ is monotone-closed if and only if the poset $([0,1]_A,\leq)$ is directed-complete.
\end{myprop}

\begin{proof}
Let $A$ be a C*-algebra.

If $A$ is monotone-closed, then, by definition every increasing net of positive elements bounded by $1$ has a least upper bound in $[0,1]_A$ and therefore, the poset $([0,1]_A,\leq)$ is directed-complete.

Conversely, suppose that $[0,1]_A$ is directed-complete. We now consider an increasing net of positive elements $(a_\lambda)_{\lambda \in \Lambda}$ in $A^+$, bounded by a nonzero positive element $b \in A^+$. Then, it restricts to an increasing net $(\frac{a_\lambda}{\norm{b}})_{\lambda \in \Lambda}$ in $[0,1]_A$ since $b \leq \norm{b} \scprod 1$. By assumption, the increasing net $(\frac{a_\lambda}{\norm{b}})_{\lambda \in \Lambda}$ has a least upper bound $\lublambda \frac{a_\lambda}{\norm{b}} \in \unit_A$ and thus $\norm{b} \lublambda \frac{a_\lambda}{\norm{b}}$ is an upper bound for $(a_\lambda)_{\lambda \in \Lambda}$.

Let $c \in A^+$ be an upper bound for the increasing net $(a_\lambda)_{\lambda \in \Lambda}$ such that $c \leq b$. For every $\lambda' \in \Lambda$, $a_{\lambda'} \leq c \leq b \leq \norm{b} \scprod 1$ and thus $\frac{c}{\norm{b}}$ is an upper bound for the increasing net $(\frac{a_\lambda}{\norm{b}})_{\lambda \in \Lambda}$. It follows that $\lublambda \frac{a_\lambda}{\norm{b}} \leq \frac{c}{\norm{b}}$ and therefore, $\norm{b} \lublambda \frac{a_\lambda}{\norm{b}} \leq c$. Thus, $\norm{b} \lublambda \frac{a_\lambda}{\norm{b}}$ is the least upper bound of the increasing net $(a_\lambda)_{\lambda \in \Lambda}$ bounded by $b$ and we can conclude that $A$ is monotone-closed.
\end{proof}

In this thesis, we have chosen to use the standard definition of normal maps. However, one can say that a PsU-map is normal if its restriction $f : [0,1]_A \to [0,1]_B$ is Scott-continuous.


\begin{myprop}\label{nrml-alt}  
A PsU-map $f : A \to B$ between C*-algebras is normal if and only if its restriction $f : [0,1]_A \to [0,1]_B$ is Scott-continuous.
\end{myprop}

\begin{proof}
Let $f : A \to B$ be a positive map between two C*-algebras $A$ and $B$.

If $f$ is normal, then by definition every increasing net $(x_\lambda)_{\lambda \in \Lambda}$ in $[0,1]_A \subseteq A^+$ with least upper bound $\lub x_\lambda \in \unit_A$ is such that the net $(f(x_\lambda))_{\lambda \in \Lambda}$ is an increasing net in $\unit_B \subseteq \B^+$ with least upper bound $\lub f(x_\lambda) = f(\lub x_\lambda) \in \unit_B$. That is to say, the restriction $f : [0,1]_A \to [0,1]_B$ is Scott-continuous.

Conversely, suppose that the restriction $f : [0,1]_A \to [0,1]_B$ is Scott-continuous. Let $(x_\lambda)_{\lambda \in \Lambda}$ be an increasing net in $A^+$ with a nonzero least upper bound $y \in A^+$. Since $y \leq \norm{y} \scprod 1$, it restricts to an increasing net $(\frac{x_\lambda}{\norm{y}})_{\lambda \in \Lambda}$ in $[0,1]_A$ with a least upper bound $\frac{y}{\norm{y}}$. From the Scott-continuity of $f : [0,1]_A \to [0,1]_B$, we deduce that the net $(f(\frac{x_\lambda}{\norm{y}}))_{\lambda \in \Lambda}$ is an increasing net in $\unit_B$ with least upper bound $\lub f (\frac{x_\lambda}{\norm{y}}) = f(\frac{y}{\norm{y}}) \in \unit_B$. It follows that the net $(f(x_\lambda))_{\lambda \in \Lambda}$, which is equal to $(\norm{y} f(\frac{x_\lambda}{\norm{y}}))_{\lambda \in \Lambda}$ by linearity, is an increasing net in $B^+$ with an upper bound $\norm{y} \lub f (\frac{x_\lambda}{\norm{y}}) = f(\norm{y}\frac{y}{\norm{y}}) = f(y) \in B^+$.

Suppose that $z \in B^+$ is an upper bound for the increasing net $(f(x_\lambda))_{\lambda \in \Lambda}$. From the fact that $f(x_{\lambda'}) \leq z$ and therefore $f(\frac{x_{\lambda'}}{\norm{y}})=\frac{f(x_{\lambda'})}{\norm{y}} \leq \frac{z}{\norm{y}}$ for every $\lambda' \in \Lambda$, we obtain that $f(\frac{y}{\norm{y}}) \leq \frac{z}{\norm{y}}$ and thus $f(y) \leq z$. It follows that $f(y)$ is the least upper bound of the increasing net $(f(x_\lambda))_{\lambda \in \Lambda}$. Hence, we can conclude that the map $f$ is normal.
\end{proof}

It is known that a C*-algebra $A$ is a W*-algebra if and only if it is monotone-complete and admits  sufficiently many normal states, i.e. the set of normal states of $A$ separates the points of $A$, see \cite[Theorem 3.16]{takesaki1}. By combining this fact and Proposition \ref{mntn-clsd-alt}, one can provide an order-theoretic characterization of W*-algebras, as in the following theorem.

\begin{mytheo}\label{thm-takesaki}
Let $A$ be a C*-algebra. 

Then $A$ is a W*-algebra if and only if its set of effects $[0,1]_A$ is directed-complete with a separating set of normal states (i.e. $\forall x \in A, \exists f \in \wPSU(A,\unit_\C), f(x) \neq 0$).
\end{mytheo}

It is natural to ask which role is played by normal states in this theorem. The existence of a separating set of normal states for every W*-algebra will be seen later in the proof of Lemma \ref{zeta-jones}. For every C*-algebra $A$, it is known that normal states induce a representation $\pi$ of $A$, i.e. a *-homomorphism from $A$ to $\B(H)$, for some Hilbert space $H$, see \cite[I.9]{takesaki1}. It can be shown that, when the C*-algebra $A$ admits a separating set of normal states, the representation $\pi$ of $A$ induced by the normal states of $A$ is faithful (i.e. injective) and that, when $A$ is monotone-closed, the image $\pi(A)$ of $A$ by the faithful representation $\pi : A \to \B(H)$ is a strongly-closed *-subalgebra of $\B(H)$, which is an alternative definition of W*-algebras, see \cite[III.3]{takesaki1} for a more detailed proof.

It is important to note that, in one of the very first articles about W*-algebras \cite{kadison-wdcpo}, Kadison defined W*-algebras as monotone-closed C*-algebras which separates the points. However, to our knowledge, this definition never became standard.

\newpage
\section{A state-and-effect triangle for discrete subprobabilistic computation}\label{appendix-sub-adj}
In this appendix, we will provide an adjunction between the category of generalized effect modules and the category of subconvex sets. Then, we will use this adjunction to express a weakest precondition calculus in terms of bijective correspondences, as seen in Section \ref{WP}.

\begin{mylem}\label{functor-sconv}
For every subconvex set $X$, the homset $\SConv(X,[0,1])$ is a generalized effect module. 

Therefore, there is a functor $\SConv(-,[0,1]) : \opp{\SConv} \to \GEMod$.
\end{mylem}

\begin{proof}
Let $X$ be a subconvex set. We define pointwise a generalized effect module structure on the homset $\SConv(X,[0,1])$.

We take the map $x \mapsto 0$ as zero element.

The sum is defined pointwise for every $x \in X$ by $(f \ovee g)(x) = f(x) + g(x)$ when $f(x)+g(x)\leq 1$. Clearly, $f \ovee g$ is again an affine map of subconvex sets: \begin{align*}
(f \ovee g)(\sum_i r_i x_i) &= f(\sum_i r_i x_i) + g(\sum_i r_i x_i) \\ &= \sum_i r_i f(x_i) + \sum_i r_i g(x_i) \\ &= \sum_i r_i (f(x_i) + g(x_i)) \\ &= \sum_i r_i (f \ovee g)(x_i)\end{align*} where $\sum_i r_i x_i \in X$.

We now need to check that the homset $\SConv(X,[0,1])$ satisfies the cancellative law and the positivity law of generalized effect algebras. Let $f,g,h \in \SConv(X,\unit)$.

\begin{description}
 \item[Cancellative law:] Suppose that $f \ovee g = f \ovee h$. From the fact that $f(x) + g(x) = (f \ovee g)(x) = (f \ovee h)(x) = f(x) + h(x)$ for every $x \in X$, we deduce that $g(x)=h(x)$ for every $x \in X$ and thus $g=h$.
 \item[Positivity law:] Suppose that $f \ovee g = 0$. It follows that $f(x)+g(x)=(f\ovee g)(x)=0$ for every $x \in X$. Since the effect algebra $\unit$ is therefore a generalized effect algebra, it must satisfy the positive law and thus $f(x)=g(x)=0$ for every $x \in X$. Hence, $f=g=0$.
\end{description}

\noindent The scalar product is also defined pointwise by $r \bullet f = \lambda x \in X. r \cdot f(x)$, which is again an affine map of subconvex sets: \begin{align*}
(r \bullet f)(\sum_i r_i x_i) &= r \cdot f(\sum_i r_i x_i) \\ &= r \cdot (\sum_i r_i f(x_i)) \\ &= \sum_i r_i \cdot r \cdot f(x_i) \\ &= \sum_i r_i (r \bullet f)(x_i)\end{align*} where $\sum_i r_i x_i \in X$ and $r \in \unit$.

Thus, the mapping $X \mapsto \SConv(X,\unit)$ gives a contravariant functor: by precomposition,  we obtain a map of generalized effect modules $(-) \circ f : \SConv(Y,[0,1]) \to \SConv(X,[0,1])$ for every affine map $f : X \to Y$ of subconvex sets.

\begin{itemize}
 \item For every affine map $f : X \to Y$ of subconvex sets, \[(\lambda y \in Y. 0) \circ f = \lambda x \in X. (\lambda y \in Y.0)(f(x)) = \lambda x \in X. 0.\]
 \item Let $f : X \to Y$ be an affine map of subconvex sets and $g_1, g_2 \in \SConv(Y,[0,1])$. Suppose that $g_1 \perp g_2$ in $\SConv(Y,[0,1])$. Then, $g_1(y)+g_2(y) \leq 1$ for every $y \in Y$. Therefore, since $f(x) \in Y$ for every $x \in X$, we obtain that $g_1(f(x))+g_2(f(x))=(g_1 \circ f)(x)+(g_2 \circ f)(x) \leq 1$ for every $x \in X$. It follows that $(g_1 \circ f) \perp (g_2 \circ f)$ in $\SConv(X,[0,1])$. Moreover, \begin{align*}(g_1 \ovee g_2) \circ f &= \lambda x. (g_1 \ovee g_2)(f(x)) \\ &= \lambda x. g_1(f(x)) + g_2(f(x)) \\ &= \lambda x. (g_1 \circ f)(x) + (g_2 \circ f)(x) \\ &= \lambda x. ((g_1 \circ f) \ovee (g_2 \circ f))(x) \\ &= (g_1 \circ f) \ovee (g_2 \circ f).\end{align*}
 
 \item Let $f : X \to Y$ be an affine map of subconvex sets, $r \in [0,1]$ and $g \in \SConv(Y,[0,1])$. Then, \begin{align*}(r \bullet g) \circ f &= \lambda x. (r \bullet g)(f(x)) \\ &= \lambda x. r \cdot g(f(x)) \\ &= \lambda x. r \cdot (g \circ f)(x) \\ &= r \bullet (g \circ f)\end{align*}.
\end{itemize}
\end{proof}

\begin{mylem}\label{functor-eas}
For every generalized effect module $E$, the homset $\GEMod(E,[0,1])$ is a subconvex set.

Therefore, there is a functor $\GEMod(-,[0,1]) : \GEMod \to \opp{\SConv}$.
\end{mylem}

\begin{proof}
Let $(E,0,\ovee)$ be a generalized effect module. Let $f : E \to [0,1]$ be the subconvex sum defined pointwise by $f(x) = \sum r_i f_i(x)$ where $f_i \in \GEMod(E,[0,1])$ and $r_i \in [0,1]$ with $\sum_i r_i \leq 1$. We will now show that $\GEMod(E,[0,1])$ is a subconvex set by proving that $f \in \GEMod(E,[0,1])$.

The preservation of zero is easy: $f(0) = \sum_i r_i f_i(0) = 0$. Let $x, y \in E$ be two elements such that $x \perp_E y$. Then, \[f(x \ovee y) = \sum_i r_i f_i(x \ovee y) = \sum_i r_i (f_i(x) + f_i(y)) = \sum_i r_i f_i(x) + \sum_i r_i f_i(y) = f(x) + f(y).\] Moreover, for $r \in \unit$ and $x \in E$,  \[f(r \bullet x) = \sum r_i f_i(r \bullet x) = \sum r_i (r \bullet f_i(x)) = r \bullet (\sum r_i f_i(x)) = r \bullet f(x).\] It follows that the map $f$ preserves the sum and the scalar product, and thus $f \in \GEMod(E,[0,1])$.

Hence, the mapping $E \mapsto \GEMod(E,[0,1])$ gives a contravariant functor: for every map $g : E \to F$ of generalized effect modules, we obtain by precomposition an affine map $(-) \circ g : \GEMod(F,[0,1]) \to \GEMod(E,[0,1])$ of subconvex sets: 

$(\sum_i r_i f_i) \circ g = \lambda x. (\sum_i r_i f_i)(g(x)) = \lambda x. \sum_i r_i f_i(g(x)) = \lambda x. \sum_i r_i (f_i \circ g)(x) = \sum_i r_i (f_i \circ g)$ where $\sum_i r_i f_i \in  \GEMod(F,[0,1])$.
\end{proof}

The combination of the previous two lemmas yields an adjunction described in the following theorem.

\begin{mytheo}\label{adjunction-sconv-emods}
There is an adjunction between $\SConv$ and $\GEMod$ by "homming into $[0,1]$":

$$\xymatrix@R+.5pc{
\opp{\GEMod}\ar@/^1.5ex/[rr]^-{\GEMod(-,[0,1])} 
   & \top & \SConv\ar@/^1.5ex/[ll]^-{\SConv(-,[0,1])} 
}$$

\end{mytheo}
 
\begin{proof}
We need to establish a counit-unit adjunction between the categories $\SConv$ and $\GEMod$ with the functor $\SConv(-,[0,1])$ from Lemma \ref{functor-sconv} and the functor $\GEMod(-,[0,1])$ from Lemma \ref{functor-eas}.

Let $E$ be an arbitrary generalized effect algebra and $X$ be an arbitrary subconvex set.

We first need to check that the unit $\eta : E \to \SConv(\GEMod(E,[0,1]),[0,1])$ defined by $\eta(x) = \lambda f \in \GEMod(E,[0,1]). f(x)$ is a map of generalized effect modules.

The preservation of zero is easy: $\eta(0)= \lambda f. f(0) = \lambda f. 0 = 0$. Furthermore, if $x \perp_E y$, then: \[\eta(x \ovee y) = \lambda f. f(x \ovee y) = \lambda f. f(x) + f(y) = \lambda f. \eta(x)(f) + \eta(y)(f) = \eta(x) + \eta(y).\] Finally, we observe for every $r \in \unit$ that \[\eta (r \bullet x) = \lambda f. \eta(r \bullet x)(f) = \lambda f. f(r \bullet x) = \lambda f. r \bullet f(x) = \lambda f. r \bullet \eta(x)(f) = r \bullet \eta (x)\].

In much the same way, we need to prove that the counit $\varepsilon : X \to \GEMod(\SConv(X,[0,1]),[0,1])$ defined by $\varepsilon(x) = \lambda f \in \SConv(X,[0,1]). f(x)$ is an affine map of subconvex sets: \[\varepsilon(\sum_i r_i x_i) = \lambda f. f(\sum_i r_i x_i) = \lambda f. \sum_i r_i f(x_i) = \lambda f. \sum r_i \varepsilon(x_i)(f) = \sum_i r_i \varepsilon(x_i).\]
\end{proof}

We can now try to establish a state-and-effect triangle for discrete subprobabilistic computations, denoted via the subdistribution monad.

We define a predicate functor $\PredS : \kl{\SDM} \to \opp{\GEMod}$ by \[\PredS(X)=\SConv(\SDM(X),\unit)\] for every set $X$. Then, we can describe the situation by a state-and-effect triangle for discrete subprobabilistic computations: $$\xymatrix@R+.5pc{
\opp{\GEMod}\ar@/^1.5ex/[rr]^-{\GEMod(-,[0,1])} 
   & \top & \SConv\ar@/^1.5ex/[ll]^-{\SConv(-,[0,1])} \\
& \kl{\SDM} \ar[ul]^{\PredS\quad\quad}\ar[ur]_{\quad\mathcal{K}}
}$$ where $\mathcal{K}$ is the standard (full and faithful) "comparison" functor from the Kleisli category of a monad in its Eilenberg-Moore category.

\newpage
\section{A state-and-effect triangle for quantum computation}\label{appendix-quantum-tri}
In this section, we will provide a state-and-effect triangle for quantum computations, in order to give a categorical interpretation of a quantum weakest pre-condition calculus.

\begin{mylem}\label{functor-ksconv}
For every subconvex set $X$, the homset $\SConv(X,[0,1])$ is in $\sdcGEMod$.

Therefore, there is a functor $\SConv(-,[0,1]) : \opp{\SConv} \to \sdcGEMod$.
\end{mylem}

\begin{proof}
Let $X$ be a convex set. The homset $\SConv(X,\unit)$ carries a generalized effect module structure defined by Lemma \ref{functor-sconv}. The order on $\SConv(X,\unit)$ is defined pointwise by $f \leq g$ if and only if $\forall x \in X, f(x) \leq g(x)$, with join of elements calculated pointwise.

Let $(f_\lambda)_{\lambda \in \Lambda}$ be a directed set in $\SConv(X,[0,1])$. It is bounded by $\lambda x \in X. 1$. Then, for every $x \in X$, since the unit interval $\unit$ is bounded-complete, the directed set $(f_\lambda(x))_{\lambda \in \Lambda}$ in $\unit$ has a least upper bound $\lublambda (f_\lambda(x)) \in \unit$. Thus, from the fact that $\forall x \in X, f_{\lambda'}(x) \leq \lublambda(f_\lambda(x))$ holds for every $\lambda' \in \Lambda$, we obtain an upper bound $f \in \SConv(X,\unit)$ for the directed set $(f_\lambda)_{\lambda \in \Lambda}$ defined pointwise by $f(x) = \lublambda (f_\lambda(x))$, where $x \in X$. It is easy to check that the map $f$ is a least upper bound for the directed set $(f_\lambda)_{\lambda \in \Lambda}$: if the map $g \in \SConv(X,\unit)$ is an upper bound for $(f_\lambda)_{\lambda \in \Lambda}$, i.e. $\forall \lambda' \in \Lambda, \forall x \in X, f_{\lambda'}(x) \leq g(x)$, then we can deduce that $\forall x \in X, f(x)=\lublambda (f_\lambda(x)) \leq g(x)$ and thus $f \leq g$. It follows that $\SConv(X,\unit)$ is a directed-complete generalized effect module.

Let $f,g \in \SConv(X,\unit)$ such that $f$ and $g$ are distincts. Then, there is (at least) one element $x \in X$ such that $f(x) \neq g(x)$. We now consider a Scott-continuous map of generalized effect modules $\phi : \SConv(X,[0,1]) \to [0,1]$ defined by $\phi(f)=f(x)$. Then, $\phi(f) \neq \phi(g)$. It follows that $\phi$ separates the elements $f$ and $g$ and therefore, $\SConv(X,\unit)$ is in $\sdcGEMod$.

We know that by precomposition, one obtains a map of generalized effect modules \[(-) \circ f : \SConv(Y,[0,1]) \to \SConv(X,[0,1])\] for every affine map $f : X \to Y$ of subconvex sets, see Lemma \ref{functor-sconv}. In order to show that the mapping $X \mapsto \SConv(X,[0,1])$ gives a contravariant functor, we only need to check that this map $(-) \circ f$ is also Scott-continuous: 

Let $f : X \to Y$ be an affine map of subconvex sets and $(g_\lambda)_{\lambda \in \Lambda}$ be a directed set in $\SConv(Y,[0,1])$. Since $\SConv(Y,[0,1])$ is directed-complete for its pointwise order, the directed set $(g_\lambda)_{\lambda \in \Lambda}$ has a least upper bound defined pointwise by $(\lublambda g_\lambda)(y)=\lublambda (g_\lambda(y))$, where $y \in Y$. In particular, $(\lublambda g_\lambda)(f(x))=\lublambda (g_\lambda(f(x)))$ for every $x \in X$. Similarly, the set $(g_\lambda \circ f)_{\lambda \in \lambda}$ in $\SConv(X,[0,1])$, bounded by $\lambda x. 1$, is directed and thus has a least upper bound since $\SConv(X,[0,1])$ is directed-complete. Then, we observe that:
\begin{align*}
(\lublambda g_\lambda) \circ f &= \lambda x \in X. (\lublambda g_\lambda)(f(x)) = \lambda x \in X. \lublambda (g_\lambda(f(x))) \\ &= \lambda x \in X. \lublambda (g_\lambda \circ f)(x) = \lublambda (g_\lambda \circ f).
\end{align*}
\end{proof}

\begin{mylem}\label{functor-sdcemods}
For every $E \in \sdcGEMod$, the homset $\sdcGEMod(E,[0,1])$ is a subconvex set.

Therefore, there is a functor $\sdcGEMod(-,[0,1]) : \sdcGEMod \to \opp{\SConv}$.
\end{mylem}

\begin{proof}
Let $E \in \sdcGEMod$. We now consider a map $f : E \to \unit$ defined by $f(x) = \sum_i r_i f_i(x)$, where $f_i \in \sdcGEMod(E,[0,1]) \subseteq \GEMod(E,\unit)$. We know by Lemma \ref{functor-eas} that the map $f$ is a map of generalized effect modules. Thus, in order to show that $\sdcGEMod(E,\unit)$ is a subconvex set, we only need to check that the map $f$ is Scott-continuous:  since all the maps $(f_i)_i$ are Scott-continuous maps of the homset $\sdcGEMod(E,\unit)$, we conclude that \[f(\lublambda x_\lambda) = \sum r_i f_i (\lublambda x_\lambda) = \sum r_i (\lublambda f_i (x_\lambda)) = \lublambda (\sum r_i f_i (x_\lambda))\] for every directed set $(x_\lambda)_{\lambda \in \Lambda}$ in $E$ with least upper bound $\lublambda x_\lambda \in E$.

Hence, the mapping $E \mapsto \sdcGEMod(E,\unit)$ gives a contravariant functor: as in Lemma \ref{functor-eas}, for a map $g \in \sdcGEMod(E,F) \subseteq \GEMod(E,F)$ where $E, F \in \sdcGEMod \subseteq \GEMod$, we obtain by precomposition an affine map of subconvex sets $(-) \circ g : \sdcGEMod(F,[0,1]) \to \sdcGEMod(E,[0,1])$.
\end{proof}

The combination of the previous two lemmas yields an adjunction described in the following theorem.

\begin{mytheo}\label{adjunction-wstar}
There is an adjunction between $\SConv$ and $\sdcGEMod$ by "homming into $[0,1]$":

$$\xymatrix@R+.5pc{
\opp{\sdcGEMod}\ar@/^1.5ex/[rr]^-{\sdcGEMod(-,[0,1])} 
   & \top & \SConv\ar@/^1.5ex/[ll]^-{\SConv(-,[0,1])} 
}$$
\end{mytheo}

\begin{proof}
We will establish a counit-unit adjunction between the categories $\SConv$ and $\sdcGEMod$ with the functors $\SConv(-,\unit)$ and $\sdcGEMod(-,\unit)$ from Lemma \ref{functor-ksconv} and Lemma \ref{functor-sdcemods}. 

Let $E \in \sdcGEMod$ and $X \in \SConv$. Since the functor $\sdcGEMod(-,\unit)$ is a restriction of the functor $\GEMod(-,\unit)$, we already know by Theorem \ref{adjunction-sconv-emods} that:

\begin{itemize}
 \item The unit $\eta : E \to \SConv(\sdcGEMod(E,[0,1]),[0,1])$ defined by \[\eta(x) = \lambda f \in \sdcGEMod(E,[0,1]). f(x)\] is a map of generalized effect modules.
 \item The counit $\varepsilon : X \to \sdcGEMod(\SConv(X,[0,1]),[0,1])$ defined by \[\varepsilon(x) = \lambda f \in \SConv(X,[0,1]). f(x)\] is an affine map of subconvex sets.
\end{itemize} 

Thus, we only need to check that the unit is Scott-continuous to establish the adjunction: for every map $f \in \sdcGEMod(E,\unit)$, for every directed set $(x_\gamma)_{\gamma \in \Gamma}$ in $E$ with least upper bound $\lubgamma x_\gamma \in E$, the directed set $(\eta(x_\gamma))_{\gamma \in \Gamma}$ in $\SConv(\sdcGEMod(E,[0,1]),[0,1]) \in \sdcGEMod$ has a least upper bound $\lubgamma \eta(x_\gamma)$, the directed set $(\eta(x_\gamma)(f))_{\gamma \in \Gamma}=(f(x_\gamma))_{\gamma \in \Gamma}$ in $\unit$ has a least upper bound $\lubgamma f(x_\gamma) \in \unit$ since $\unit$ is bounded-complete and \begin{align*}
\eta(\lubgamma x_\gamma) &= \lambda f. \eta(\lubgamma x_\gamma)(f) \\
 &= \lambda f. f(\lubgamma x_\gamma) \\
 &= \lambda f. \lubgamma f(x_\gamma) \\
 &= \lambda f. \lubgamma \eta(x_\gamma)(f) \\
 &= \lubgamma \eta(x_\gamma).
\end{align*}
\end{proof}

We now consider a "normal state functor" $\NS : \opp{(\wPSU)} \to \SConv$ defined by:
 \begin{align*}
   \NS(A) & = \wPSU(A,\C) \simeq \sdcGEMod(\unit_A,\unit_\C) \\
   \NS(A \overset{f}{\to} B) & = (-) \circ f : \NS(B) \to \NS(A) \\
 \end{align*}

\begin{mylem}\label{zeta-jones}
For each W*-algebra $A$, there is an isomorphism $\unit_A \simeq \SConv(\NS(A),\unit)$.
\end{mylem}

\begin{proof}
For every C*-algebra $A$, we denote by $A'$ the dual space of $A$, i.e. the set of all linear maps $\phi : A \to \C$. It is known that a C*-algebra $A$ is a W*-algebra if and only if there is a Banach space $A_*$, called pre-dual of $A$, such that $(A_*)'=A$, see \cite[Definition 1.1.2]{sakai}.

We now consider the map $\zeta_X : X \to X''$ defined by $\zeta_X(x)(\phi)=\phi(x)$ for $x \in X$ and $\phi \in X'$. Let $A$ be a W*-algebra. We observe that $\zeta_{A_*} : A_* \to A'$ is a "canonical embedding" of $A_*$ into $A'$ and it can be proved that $A_*$ is a linear subspace of $A'$ generated by the normal states of $A$, i.e. $\zeta_{A_*}(A_*)=\linspan(\NS(A))$, see \cite[Theorem 1.13.2]{sakai}. Then, we can now consider the induced surjection $\zeta_{A_*} : A_* \to \linspan(\NS(A))$, which turns out to be injective (and thus bijective): for every pair $(x,y) \in A_* \times A_*$ such that $x \neq y$, there is a $f \in \NS(A)$ such that $\zeta_{A_*}(x)(f) = f(x) \neq f(y) = \zeta_{A_*}(y)(f)$, which implies that $\zeta_{A_*}(x) \neq \zeta_{A_*}(y)$.

From $A_* \xrightarrow[\simeq]{\zeta_{A_*}} \linspan(\NS(A))$ for every W*-algebra $A$, we obtain that \[\unit_A \subseteq A = (A_*)' \xrightarrow[\simeq]{\zeta_{A}} \linspan(\NS(A))' \supseteq \SConv(\NS(A),\unit)\] for every W*-algebra $A$. We can now show that $\unit_A \simeq \SConv(\NS(A),\unit)$ for every W*-algebra $A$.

Let $a \in [0,1]_A$. Then, for every $\varphi \in \linspan(\NS(A))$, $\zeta_A(a)(\varphi)=\varphi(a) \leq \varphi(1) \leq 1$ by Lemma \ref{presrv-order}. Thus, we can conclude that $\zeta_A(a) \in \SConv(\NS(A),\unit)$ : for every $\sum_i r_i \varphi_i \in \NS(A)$, \[\zeta_A(a)(\sum_i r_i \varphi_i) = \sum_i r_i \varphi_i(a) = \sum r_i \zeta_A(a)(\varphi_i).\]

Conversely, suppose that $\zeta_A(a) \in \SConv(\NS(A),\unit)$ where $a \in A$. Then, by \cite[Theorem 4.3.4(iii)]{kadison-ringrose1}, from the fact that for every $\varphi \in \NS(A)$, $\zeta_A(a)(\varphi) = \varphi(a) \in \unit$, we can conclude that $a \in \unit_A$.
\end{proof}

\begin{myprop}\label{pred-wstar}
There is a full and faithful functor $[0,1]_{(-)} : \wPSU \to \sdcGEMod$.
\end{myprop}

\begin{proof}
By Lemma \ref{psu-determined}, a PsU-map $f : A \to B$ between C*-algebras is completely determined and defined by its action on $[0,1]_A$. Moreover, by combining Theorem \ref{thm-takesaki} with Proposition \ref{nrml-alt}, we observe that :
\begin{itemize}
 \item for every C*-algebra $A$, $\unit_A \in \sdcGEMod$ if and only if $A \in \wPSU$.
 \item for every PsU-map $f : A \to B$ between C*-algebras, $f : A \to B$ is in $\wPSU(A,B)$ if and only if its restriction $f : \unit_A \to \unit_B$ is in $\sdcGEMod(\unit_A,\unit_B)$.
\end{itemize}

That is to say, there is a full and faithful functor $[0,1]_{(-)} : \wPSU \to \sdcGEMod$.

\end{proof}

\begin{myprop}
The functor $\sdcGEMod(-,\unit) : \opp{\sdcGEMod} \to \SConv$ is faithful.
\end{myprop}

\begin{proof}
Let $X, Y \in \sdcGEMod$ and $f, g \in \sdcGEMod(X,Y)$. We now suppose that $\phi \circ f = \phi \circ g$ for every $\phi \in \sdcGEMod(Y,\unit)$, which means that $\phi(f(x))=\phi(g(x))$ for every $x \in X$. Since $\sdcGEMod(Y,[0,1])$ is a separating set for $f(X) \subseteq Y$, it follows that $f(x)=g(x)$ for every $x \in X$, i.e. $f=g$.
\end{proof}

Since the functor $\NS$ is the composition of the functor $\unit_{(-)}$ by the functor $\sdcGEMod(-,\unit)$, we obtain the following result.

\begin{mycoro}
The functor $\NS : \opp{(\wPSU)} \to \SConv$ is faithful.
\end{mycoro}

The previous results give rise to the following theorem.

\begin{mytheo}
The following state-and-effect triangle is a commutative diagram: 
$$\xymatrix@R+.5pc{
\opp{(\sdcGEMod)}\ar@/^1.5ex/[rr]^-{\sdcGEMod(-,[0,1])} 
   & \top & \SConv\ar@/^1.5ex/[ll]^-{\SConv(-,[0,1])} \\
& \opp{(\wPSU)} \ar[ul]^-{[0,1]_{(-)}}\ar[ur]_{\quad\NS}
}$$
\end{mytheo}

\newpage
\section{Domain-theoretic properties of the lattices of projections on Hilbert spaces}\label{appendix-lattices}
In this section, after recalling some standard definitions of lattice theory and domain theory, we will define a special class of operators known as projections that play a crucial role in operator theory since von Neumann and Birkhoff proposed in \cite{birkhoff-vonNeumann} to use projections to represent mathematically the properties of physical systems.

\begin{mydef}
Let $P$ be a poset. For elements $x$ and $y$ in $P$, one says that $x \ll y$ ("x approximates y" or "x is way below y") if for any directed set $\Delta \dirsubset P$, $y \leq \bigvee \Delta$ implies that there is a $d \in \Delta$ such that $x \leq d$.
\end{mydef}

\begin{mydef}
Let $D$ be a dcpo.

An element $x \in D$ is compact if $x \ll x$. We denote by $\Komp(D)$ the set of compact elements of $D$. 

$D$ is called algebraic or an algebraic domain if every element of $D$ is the least upper bound of the compact elements below it, i.e. for every $x \in D$, ${\downarrow} x \cap \Komp(D) = \set{y \in \Komp(D) \mid y \leq x}$ is directed with $x$ as least upper bound. 




\end{mydef}

\begin{mydef}
A lattice is a poset ($L$,$\leq$) in which every pair of elements $(a,b)$ has a meet $a \wedge b$ and a join $a \vee b$ such that for every element $c$ of $L$:
\begin{itemize}
 \item $a \wedge b \leq a, b$ and $a, b \leq a \vee b$.
 \item $c \leq a, b$ implies $c \leq a \wedge b$.
 \item $a, b \leq c$ implies $a \vee b \leq c$.
\end{itemize}
It follows by induction that a lattice has all the non-empty finite joins and meets of its elements.

A lattice is complete if all its subsets have both a least upper bound and a greatest lower bound.


An algebraic lattice is a complete lattice, which is algebraic as a dcpo.
\end{mydef}

It should be noted that every complete lattice is directed-complete \cite{davey-priestley}.

\begin{mydef}
Let $(P,\leq)$ be a poset with least element $0$.

An element $a \in P$ is an atom if $0 < a$ and there is no $x \in P$ such that $0 < x < a$. We denote by $\Atom(P)$ the set of atoms of $P$. 

$P$ is atomic if for every nonzero element $b \in P$, there is an atom $a \in P$ such that $0 < a \leq b$. 

An atomistic lattice is an atomic lattice $L$ where every nonzero element $x \in L$ is a join of atoms below $x$.
\end{mydef}

\begin{mydef}
Let $H$ be a Hilbert space and $P \in \B(H)$.

$P$ is a projection in $\B(H)$ if it is self-adjoint and idempotent, i.e. $P=P^*=P^2$.


We denote by $\Proj(H)$ the set of projections in a Hilbert space $H$.
\end{mydef}

It can be shown that there is a one-to-one correspondance between closed subspaces and projections in a Hilbert space $H$, see \cite[I.5.1]{blackadar}. Moreover, for every Hilbert space $H$, the set of projections $\Proj(H)$ forms an atomistic lattice, since the projections corresponding to one-dimensional subspaces are atoms\footnote{Such projections are usually called minimal in the literature, see \cite{blackadar,dixmier,piron,sakai,takesaki1}} and every closed subspace of $H$ is the closure of the span of its one-dimensional subpaces, see \cite{piron}. 

One could think that, since every lattice of projections on a Hilbert space is a complete lattice, it might be possible to provide an algebraic lattice of projections as a mathematical model for quantum computability analysis like in \cite{edalat-heckmann}. It turns out that it is not possible for every Hilbert space, as shown by the following counter-example.

\begin{mysam}
Consider $\ltwo$ the set of countable square-summable sequences of complex numbers, i.e. the set of infinite sequences $(c_1,c_2,\cdots)$ of complex numbers where $c_i \in \C$ such that the sum $\sum_i \abs{c_i}$ is finite. It is known that $\ltwo$ is an Hilbert space with the inner product $\prodscal{x}{y} = \sum_{n \in \N} x_n \overline{y_n}$. This Hilbert space has an orthonormal basis $(e_i)_{i \in \N}$ defined by $e_1 = (1,0,0,\cdots)$, $e_2 = (0,1,0,\cdots)$, ...

Let $(p_n)_{n \geq 1}$ be a directed set of projections in $\ltwo$ such that $p_1$ corresponds to the closed subspace spanned by $e'_1=e_1$ and $p_n$ corresponds to the closed subspace spanned by $e'_n=\frac{1}{n} e_n + e_1$ for every $n \geq 2$.

Then, $\lim_{n \to \infty}(e'_n)=e_1$ implies the range $\ran \lub p_n$ contains elements of $\ltwo$ that are arbitrary close to $e_1$ and hence contains $e_1$. For every $m \geq 2$, from $e_m = m(e'_m - e_1)$, we deduce that $e_m$ is contained in the closed subspace corresponding to $p_m$ and therefore in the closed subspace corresponding to $\lub p_n$.

Let $(p'_n)_{n \geq 2}$ be the directed set defined by $p'_n = p_2 \vee \cdots \vee p_n$ for every $n \geq 2$. Then, we observe that, $p_m \leq p'_m \leq \lub_{n \geq 2} p'_n$ for every $m \geq 2$ although there is no $m \geq 2$ such that $p_1 \leq p'_m$. Thus, $p_1$ is not way below $p_n$ ($n \geq 2$) and by symmetry, there is no way below ordering between two non-zero projections.
\end{mysam}

\end{appendices}

\newpage
\section*{Proof of Theorem \ref{thm-takesaki}}
\begin{mydef}
Let $A$ be a C*-algebra. 
A representation $\pi$ of $A$ is a *-homomorphism $\pi : A \to \B(H)$ for some Hilbert space $H$.
A representation is called faithful when it is injective.
\end{mydef}

\begin{myprop}
Every C*-algebra $A$ admits a faithful representation.
\end{myprop}

\begin{myprop}
A C*-algebra $A$ is a W*-algebra if and only if there is a faithful representation $\pi : A \to \B(H)$, for some Hilbert space $H$, such that $\pi(A)$ is a strongly-closed subalgebra of $\B(H)$.
\end{myprop}

\begin{mytheo}
Let $A$ be a C*-algebra. Then, $A$ is a W*-algebra if and only if $A$ is monotone-complete and $\NS(A)$ is a separating set for $A$.
\end{mytheo}

\begin{proof}
Let $A$ be a C*-algebra.

Suppose that $A$ is a W*-algebra. Then, by Corollary \ref{effects-ccpo}, $\unit_A$ is a dcpo and thus $A$ is monotone-complete by Proposition \ref{mntn-clsd-alt}. Moreover, we know by \cite[Theorem 1.13.2]{sakai} and Lemma \ref{zeta-jones} that there is an isomorphism $\zeta_{A} : A \to \linspan(\NS(A))'$ defined by $\zeta_A(a)(\varphi) = \varphi(a)$ for $a \in A$ and $\varphi \in A'$. Therefore, $\zeta_A$ is injective and thus for every pair $(x,y)$ of distinct elements of $A$, $\zeta_A(x) \neq \zeta_A(y)$, which means that there is a $\varphi \in \NS(A)$ such that $\varphi(x) = \zeta_A(x)(\varphi) \neq \zeta_A(y)(\varphi) = \varphi(y)$. It follows that the set $\NS(A)$ is a separating set for $A$.

Conversely, suppose that $A$ is monotone-closed and admits its normal states as a separating set.

There is a representation $\pi : A \to B(H)$, for some Hilbert space $H$, induced by the normal states on $A$, by the Gelfand-Naimark-Segal (GNS) construction \cite[Theorem I.9.14, Definition I.9.15]{takesaki1}: 
\begin{itemize}
 \item Every normal state $\omega$ on $A$ induce a representation $\pi_\omega : A \to \B(H_\omega)$ such that there is a vector $\xi_\omega$ such that $\omega(x)=\prodscal{\pi_\omega(x)\xi_\omega}{\xi_\omega}$ for every $x \in A$
 \item We define a Hilbert space $H$, which is the direct sum of the Hilbert spaces $H_\omega$, where $\omega$ is a normal state on $A$.
 \item The representation $\pi : A \to \B(H)$ is defined pointwise for every $x \in A$: $\pi(x)$ is the bounded operator on $H$ defined as the direct sum of the bounded operators $\pi_\omega(x)$ on $H_\omega$, where $\omega$ is a normal state on $A$.
\end{itemize}

By assumption, the set of normal states of $A$ is a separating set for $A$ and thus, for every pair of distincts elements $x,y$ in $A$, there is a state $\rho$ on $A$ such that $\prodscal{\pi_\rho(x)\xi_\rho}{\xi_\rho}=\rho(x) \neq \rho(y)=\prodscal{\pi_\rho(y)\xi_\rho}{\xi_\rho}$ and thus $\pi_\rho(x) \neq \pi_\rho(y)$ for some state $\rho$ on $A$. It follows that $\pi(x) \neq \pi(y)$ and hence, the representation $\pi$ is faithful.
 
Let $\rho$ be a normal state on $A$. Since $A$ is monotone-closed, every directed set $(\rho(x_\lambda))_{\lambda \in \Lambda}$ in $\B(H)$ has a least upper bound $\lublambda \rho(x_\lambda) = \rho(\lublambda x_\lambda)$. According to the definition we gave earlier of $\pi_\rho$, this imply that $\pi_\rho(x_\lambda)$ converges weakly to $\lublambda \pi_\rho(x_\lambda)$. Since a bounded net of positive operators converges strongly whenever it converges weakly (see \cite[I.3.2.8]{blackadar}), it turns out that $\lublambda \pi_\rho(x_\lambda)$ is the strong limit of $(\pi_\rho(x_\lambda))_{\lambda \in \Lambda}$ in $\B(H_\rho)$. Hence, the strong limit of $(\pi(x_\lambda))_{\lambda \in \Lambda}$ in $\B(H)$ exists in $\B(H)$ and is defined as the direct sum of the strong limit of the nets $(\pi_\omega(x_\lambda))_{\lambda \in \Lambda}$ where $\omega$ is a normal state on $A$. Thus, $\pi(A)$ is strongly closed in $\B(H)$ and thus $A$ is a W*-algebra.
\end{proof}

\end{document}